\documentclass[12pt]{article}
\usepackage{graphicx} 
\usepackage{algorithm,algcompatible}
\usepackage{float}    
\usepackage{verbatim} 
\usepackage{amsmath}  
\usepackage{amssymb}  
\usepackage{subfig}   
\usepackage[colorlinks=true,citecolor=blue,linkcolor=blue,urlcolor=blue]{hyperref}
\usepackage{fullpage}
\usepackage{amsthm,mathtools}
\usepackage{enumerate}
\usepackage{paralist}
\usepackage{xspace}
\usepackage{caption}
\usepackage{bbm}
\usepackage{url}
\usepackage{mathrsfs}

\newcommand{\vertiii}[1]{{\left\vert\kern-0.25ex\left\vert\kern-0.25ex\left\vert #1 
		\right\vert\kern-0.25ex\right\vert\kern-0.25ex\right\vert}}

\makeatletter

\newtheorem{theorem}{Theorem}
\newtheorem{definition}{Definition}
\newtheorem{lemma}{Lemma}

\newtheorem{remark}{Remark}

\newtheorem{proposition}{Proposition}

\newtheorem{assumption}{Assumption}

\usepackage[margin=1in]{geometry}
\usepackage{times}

\usepackage{thmtools}
\usepackage{thm-restate}

\usepackage{listings}

\@ifundefined{showcaptionsetup}{}{%
 \PassOptionsToPackage{caption=false}{subfig}}
\usepackage{subfig}
\makeatother

\usepackage{listings}

\global\long\def\H2{\mathcal{H}_2}

\global\long\def\E1{\mathcal{E}_1}
\global\long\def\mA{\mathcal{A}}
\global\long\def\mAj{\mathcal{A}_j}

\begin{document}

\title{\bf Learning Sparse Dynamical Systems from a Single Sample Trajectory}
\author{Salar Fattahi,  Nikolai Matni, Somayeh Sojoudi
\thanks{Salar Fattahi is with the Department of Industrial Engineering and Operations Research, University of California, Berkeley. Nikolai Matni is with the Department of Electrical Engineering and Computer Sciences, University of California, Berkeley. Somayeh Sojoudi is with the Departments of Electrical Engineering and Computer Sciences and Mechanical Engineering as well as the Tsinghua-Berkeley Shenzhen Institute, University of California, Berkeley. This work was supported by the ONR Award N00014-18-1-2526, NSF Award 1808859 and AFSOR Award FA9550-19-1-0055.}}
\date{}
\maketitle

\begin{abstract}
This paper addresses the problem of identifying sparse linear time-invariant (LTI) systems from a single sample trajectory generated by the system dynamics.  We introduce a Lasso-like estimator for the parameters of the system, taking into account their sparse nature. Assuming that the system is stable, or that it is equipped with an initial stabilizing controller, we provide sharp finite-time guarantees on the accurate recovery of both the sparsity structure and the parameter values of the system. In particular, we show that the proposed estimator can correctly identify the sparsity pattern of the system matrices with high probability, provided that the length of the sample trajectory exceeds a threshold. Furthermore, we show that this threshold scales polynomially in the number of nonzero elements in the system matrices, but logarithmically in the system dimensions --- this improves on existing sample complexity bounds for the sparse system identification problem. We further extend these results to obtain sharp bounds on the $\ell_{\infty}$-norm of the estimation error and show how different properties of the system---such as its stability level and \textit{mutual incoherency}---affect this bound. Finally, an extensive case study on power systems is presented to illustrate the performance of the proposed estimation method.  
\end{abstract}

\section{Introduction}

Modern cyber-physical systems, such as power grids, autonomous transportation systems, and distributed computing and sensing networks, are characterized by being large scale, spatially distributed, and by having complex ever changing dynamics and interconnected topologies.  The distributed optimal control literature addresses set-point tracking and regulation in the distributed setting by assuming known dynamics with a sparse interconnections.  Indeed, the underlying sparsity structure of a distributed system is aggressively (and necessarily) exploited, with foundational results showing that both tractability~\cite{rotkowitz2005characterization} and scalability~\cite{wang2016system, wang2018separable, kheirandishfard2018convex, fattahi2019transformation} in controller synthesis are only possible when the underlying dynamical system is suitably sparse.  However, in this large-scale, dynamic, and complex setting, it is unclear how to obtain the necessary models of the dynamical systems.  To address this issue, we use data-driven approaches to identify both the interconnected topology and the dynamic behavior of these systems for which first-principle modeling becomes either intractable or impractical for such large-scale dynamic systems.

This then raises a more fundamental question: how can data-driven methods be appropriately integrated into safety-critical control loops?   This question has been addressed in the context of learning \cite{sarkar2018fast, pereira2010learning}, and control of a small-scale and dense unknown systems, e.g., a single autonomous vehicle or robot~\cite{dean2018regret, dean2017sample, dean2018safely, abbasi2011regret, faradonbeh2018finite}.  These works make clear that if a learned model is to be integrated into a safety-critical control loop, then it is essential that the uncertainty associated with the learned model be explicitly quantified. This way, the learned model and the uncertainty bounds can be integrated with tools from robust control to provide strong guarantees of system performance and stability.  This paper takes a first step towards extending these results to the large-scale distributed setting by providing a sample efficient and computationally tractable algorithm for the identification of sparse dynamical systems, as well as providing sharp estimates on the corresponding model uncertainty. 

\textbf{Main contributions:} 
We show that large-scale sparse system models can be identified with a complexity scaling quadratically with the number of nonzero elements in the underlying dynamical system---for systems composed of a large number of subsystems that only interact with a small number of local neighbors, this computational saving can be significant.  We further provide sharp bounds on the corresponding model uncertainty, paving the way for the use of these models in safety-critical control loops.  Finally, in contrast to previous work, we show that such models can be extracted from a single trajectory of the system.  In the context of large-scale systems, the system resets needed by methods relying on independent trajectories become prohibitively more expensive and impractical---indeed contrast resetting a robotic arm and a power distribution network, and the increase in difficulty becomes apparent.  Note that we defer a detailed comparison of our results to prior work to Section \ref{sec:main}.

\textbf{Paper organization:} In Section \ref{sec:problem}, we formally define the sparse system-identification task that we consider, and introduce our Lasso-like estimator based on a single system trajectory.  Section \ref{sec:main} presents our main result, and compares and contrasts it with existing results in the literature.  We also show that some of the technical assumptions that we make are necessary for a well-posed problem. We provide an empirical study of our method on a power system in Section \ref{sec:example}.  We end with conclusions in Section \ref{sec:conclusion}. The proofs are deferred to the appendix to streamline the presentation.

\noindent{\bf Notation:} For a matrix $M$, the symbols $\vertiii{M}$, $\vertiii{M}_\infty$, $\|M\|_F$, $\|M\|_1$, and $\|M\|_\infty$ are used to denote its induced spectral, induced infinity, Frobenius, element-wise $\ell_1/\ell_1$, and element-wise $\ell_\infty/\ell_\infty$ norms, respectively. Furthermore, $\|M\|_0$ refers to the number of nonzero elements in $M$. The symbols $M_{:j}$ and $M_{j:}$ indicate the $j^{\text{th}}$ column and row of $M$, respectively. For a set $\mathcal{I}$, the symbol $|\mathcal{I}|$ denotes its cardinality. Given the index sets $\mathcal{U}$ and $\mathcal{V}$, define $M_{\mathcal{U}\mathcal{V}}$ as the $|\mathcal{U}|\times |\mathcal{V}|$ submatrix of $M$ obtained by removing the rows and columns with indices not belonging to $\mathcal{U}$ and $\mathcal{V}$. The symbols $c$ and $c_i$ play the role of universal constants throughout the paper. $\mathbb{E}\left\{x\right\}$ denotes the expected value of a random variable $x$. For an event $\mathcal{E}$, the notation $\mathbb{P}(\mathcal{E})$ refers to its probability of occurrence. The notation $x_n\overset{a.s.}{\rightarrow} x$ means that a sequence of random variables $x_n$ converges to $x$ almost surely.

\section{Problem Statement}\label{sec:problem}

Consider the linear time-invariant (LTI) system
\begin{align}\label{ss}
& x(t+1) = Ax(t) + Bu(t) + w(t)
\end{align}
where $A\in\mathbb{R}^{n\times n}$ and $B\in\mathbb{R}^{n\times m}$ are the unknown state and input matrices, respectively. Furthermore, $x(t)\in\mathbb{R}^n$, $u(t)\in\mathbb{R}^m$, and $w(t)\in\mathbb{R}^n$ are the respective state, input, and disturbance vectors at time $t$.

The goal of this work is to estimate the underlying parameters of the dynamics, based on a limited number of \textit{sample trajectories}, i.e., a sequence $\{(x^{(i)}(\tau),u^{(i)}(\tau))\}_{\tau = 0}^T$ with $i = 1,2,...,d$, where $d$ is the number of available sample trajectories and $T$ is the length of each sample trajectory. To simplify the notations, the superscript $i$ is dropped from the sample trajectories when $d=1$. 

This paper is concerned with the identification of high dimensional but sparse system matrices $(A,B)$.  Such high-dimensional sparse parameters arise in the context of large-scale distributed and multi-agent systems, where dynamic coupling arises due to local interactions between subsystems--it is this local interaction structure that results in correspondingly sparse system matrices.  Examples of such systems include power grids, intelligent transportation systems, and distributed computation and sensing networks.

We now compare and contrast two approaches to collecting sample trajectories from a dynamical system~\eqref{ss}:

\vspace{2mm}
\noindent\textbf{Fixed $\mathbf{d}$ and variable $\mathbf{T}$:} In this method, the number of sample trajectories $d$ is set to a fixed value (e.g., $d = 1$) and instead, a sufficiently long time horizon (also referred to as learning time) $T$ is chosen to collect enough information about the dynamics. This approach is most suitable when the open-loop system is stable, or if a stabilizing controller is provided---note that this assumption of stability is necessary, as even a simple least-squares estimator may not be consistent if the system has unstable modes~\cite{sarkar2018fast}. From a practical perspective, system instability may also impose limits on how large the learning time can be in order to ensure system safety, thereby restricting the amount of data that can be collected.

\vspace{2mm}
\noindent\textbf{Fixed $\mathbf{T}$ and variable $\mathbf{d}$:} In this approach, the learning time $T$ is fixed and instead, the number of sample trajectories is chosen to be sufficiently large. Notice that this method is not dependent on the system stability. However, one needs to reset the initial state of the system at the beginning of each sample trajectory, which may not be possible in practice, especially in the case of large-scale systems.
\vspace{2mm}

This work focuses on \textit{sparse} system identification using a single trajectory, where it is assumed that the system is either stable, or equipped with an initial stabilizing controller, and our goal is to both identify the supports of the sparse system matrices $(A,B)$ and estimate their values, using a single sample trajectory. As mentioned in~\cite{dean2018regret}, in many applications, the existence of an initial stabilizing controller for the unknown system~\eqref{ss} is not restrictive. In fact,~\cite{dean2017sample} and~\cite{faradonbeh2018finite} respectively introduce offline and adaptive procedures for designing such an initial stabilizing controller.


Indeed, one can cast the sparse system identification task as a \textit{supervised learning} problem, where the goal is to fit the linear model~\eqref{ss}---parameterized by $(A,B)$---to a limited number of measurements $\{(x(\tau),u(\tau))\}_{\tau = 0}^T$. Motivated by this observation, one can consider the following $M$-estimator:
\begin{align}\label{ls_single}
(\hat{A},\hat{B}) = \arg\min_{A, B}\ &\frac{1}{2T}\sum_{t=0}^{T-1} \left\|x(t+1)-\left(Ax(t)+Bu(t)\right)\right\|_2^2\nonumber\\
&+ \lambda(\|A\|_1+\|B\|_1).
\end{align}
where the first term corresponds to the maximum likelihood estimation of $(A,B)$ when the disturbance noise has a zero-mean Gaussian distribution, and the second term has the role of promoting sparsity in the estimated $(\hat A,\hat B)$. 

Before proceeding, it is essential to note that there are fundamental limits on the performance of the introduced estimator. In particular, the above optimization problem may not have a unique solution for any length of the sample trajectory. To see this, suppose that $u(t) = K_0x(t)$ and $K_0$ is equal to the identity matrix. Then, the above optimization problem reduces to
\begin{align}
(\hat{A},\hat{B}) = \arg\min_{A, B}&\frac{1}{2T}\sum_{t=0}^{T-1} \left\|x(t+1)-\left(A+B\right)x(t)\right\|_2^2\nonumber\\
&+ \lambda(\|A\|_1+\|B\|_1).\nonumber
\end{align}
It is easy to see that, given any optimal solution $(\hat{A},\hat{B})$ to the above optimization, $(\tilde{A}, \tilde B) = (\alpha\hat{A},(1-\alpha)\hat{B})$ is also optimal for any $0\leq\alpha\leq 1$. To break this symmetry and to guarantee the identifiability of the parameters, it is essential to inject an \textit{input noise} to the system at every time $t$. In particular, we assume that $u(t) = K_0x(t)+v(t)$, where $v(t)$ is a random vector with a user-defined distribution. As another example, if $A$ is stable and $K_0 = 0$, the need to introduce noise in the input is inevitable in order to identify the matrix $B$.

To further analyze the properties of the above estimator, one can write~\eqref{ss} in a compact form. Let $\Psi^* = \begin{bmatrix}
A & B
\end{bmatrix}^\top$ denote the true parameters of the system. Furthermore, define
\begin{align}
& Y \!=\! \begin{bmatrix}
x(1)^\top\\
\vdots\\
x(T)^\top
\end{bmatrix}, X \!=\! \begin{bmatrix}
x(0)^\top & u(0)^\top\\
\vdots & \vdots\\
x(T\!-\!1)^\top & u(T\!-\!1)^\top
\end{bmatrix},
W \!=\! \begin{bmatrix}
w(0)^\top\\
\vdots\\
w(T\!-\!1)^\top
\end{bmatrix}.
\end{align}
The system identification problem is then reduced to estimating the unknown parameter $\Psi^*$ given the \textit{design matrix} $X$, and the \textit{observation matrix} $Y$ that is corrupted with the \textit{noise matrix} $W$. We can therefore rewrite optimization problem~\eqref{ls_single} compactly as
\begin{equation}\label{ls2_single}
\hat{\Psi} = \arg\min_{\Psi}\frac{1}{2T}\|Y-X\Psi\|_F^2+\lambda\|\Psi\|_1
\end{equation} 
which corresponds to the so-called \textit{Lasso} estimator, initially popularized in statistics and machine learning to estimate the support parameter values of a sparse linear model~\cite{tibshirani1996regression}. The non-asymptotic properties of this estimator have been widely studied in the literature~\cite{wainwright2009sharp, meinshausen2006high, zhao2006model}, all highlighting its sub-linear sample complexity under suitable technical conditions. In particular, they show that under the so-called \textit{mutual incoherency} of the design matrix and the sparsity of the unknown parameters, the minimum number of observations for the accurate estimation of the Lasso scales logarithmically in the dimension of $\Psi$. Motivated by these results, one may speculate that the proposed estimator~\eqref{ls_single} benefits from a similar logarithmic sample complexity. However, the validity of the derived non-asymptotic estimation error bounds on the Lasso is contingent upon a number of assumptions on the independence between the design matrix $X$ and the noise matrix $W$~\cite{wainwright2009sharp, negahban2012unified}; such assumptions do not necessarily hold in the sparse system identification problem, partly due to the dependency between the states, the inputs and the disturbance noise. The problematic nature of this dependency becomes more evident by noting that the Lasso may not be consistent when the design and noise matrices are dependent~\cite{fan2014endogeneity}.

%

This lack of independence in the design and noise matrices of the sparse system identification problem has been the main roadblock in deriving similar sub-linear sample complexity bounds for the sparse system identification problem and it leaves the following question unanswered:

\vspace{2mm}
{\it Is the estimator~\eqref{ls_single} consistent, and if so, what is its sample complexity?}

\section{Main Results}\label{sec:main}

Despite the fact that in general, the Lasso may not be a consistent estimator when the design and noise matrices are dependent, we exploit the underlying structure of the system identification problem to control this dependency and provide an affirmative answer to the posed question. In other words, we show that not only is the proposed estimator~\eqref{ls_single} consistent, but that it also enjoys a logarithmic sample complexity in the state and input dimensions, under appropriate conditions. To this goal, we first provide a number of definitions.
\begin{definition}
	A zero-mean (centered) random variable $x$ is \textbf{sub-Gaussian} with parameter $b$ if its moment generating function satisfies
	\begin{align}\nonumber
	\mathbb{E}\{\exp(tx)\}\leq\exp\left(\frac{b^2t^2}{2}\right)
	\end{align}
	for every $t$.
\end{definition}
For a centered sub-Gaussian random variable $x$ with parameter $b$, one can easily verify that $\mathbb{P}(|x|>t)\leq2\exp\left(\frac{t^2}{2b^2}\right)$. 
The most commonly known examples of such random variables are Gaussian, Bernoulli, and any bounded random variable.
\begin{definition}
	Given a sub-Gaussian random variable $x$, its \textbf{sub-Gaussian norm}, denoted by $\|x\|_{\psi}$ is defined as the smallest $r>0$ such that the inequality $\mathbb{E}\{{x^2}/{r^2}\}\leq 2$ is satisfied.
\end{definition}
It is well-known that the above two definitions are closely related. In particular, it can be verified that $\frac{1}{\sqrt{5}}b\leq \|x\|_{\psi}\leq\sqrt{\frac{8}{3}}b$ for a sub-Gaussian random variable with parameter $b$.\footnote{This is a standard result; see~\cite{rivasplata2012subgaussian} and~\cite{wainwright2019high} for a simple proof.} For a random vector $x$ with sub-Gaussian elements, $\|x\|_\psi$ is defined as $\max_i\{\|x_i\|_\psi\}$.

As mentioned before, we assume that the dynamical system is equipped with an initial static and stabilizing state-feedback controller $K_0$. More specifically, we assume that at any given time $t$, the input $u(t)$ is equal to $K_0x(t)+v(t)$, where $v(t)$ is a user-defined {input noise} with independent and centered sub-Gaussian elements whose non-zero variance is upper bounded by $\sigma_v^2$ (for stable systems, $K_0$ can be set to zero). Similarly, we assume that the disturbance noise at every time $t$ is a random vector with independent and centered sub-Gaussian elements whose variance is upper bounded by $\sigma_u^2$. Further, let $\eta>0$ be the smallest positive constant such that $\max\{\|w(t)\|_\psi,\|v(t)\|_\psi\}\leq\eta$; such a constant is guaranteed to exist as $w$ and $v$ are assumed to be centered sub-Gaussian random variables.
\begin{remark}
	Most of the existing results on the sample complexity of the system identification problem assume a centered Gaussian distribution for the input noise~\cite{pereira2010learning, Salar18, dean2017sample}. Despite having desirable finite-time properties, these types of Gaussian inputs may jeopardize the safety of the dynamical system due to their unbounded range. Accordingly, in many control systems, the input is constrained to have a limited power. These types of constraints can be translated into $\ell_{\infty}$ or $\ell_2$ bounds on the input signal. Due to the fact that such bounded random signals are sub-Gaussian, our results are readily applied to system identification problems with input constraints.
\end{remark}
Notice that for LTI systems, the uniform asymptotic stability of the closed-loop system is equivalent to its exponential stability. In other words, an LTI system is uniformly asymptotically stable if and only if there exist constants $C\geq1$ and $0<\rho<1$ such that $\vertiii{(A+BK_0)^\tau}\leq C\rho^\tau$ for every time $\tau$. Without loss of generality, let $C\geq 1$ and $0\leq \rho <1$ be the smallest constants such that $\vertiii{(A+BK_0)^\tau B}\leq C\rho^\tau$, $\vertiii{K_0(A+BK_0)^\tau}\leq C\rho^\tau$ and $\vertiii{K_0(A+BK_0)^\tau B}\leq C\rho^\tau$ for every time $\tau$. Note that the existence of such $C\geq 1$ and $0<\rho<1$ is guaranteed due to the exponential stability of the closed-loop system.

Furthermore, we assume that the initial state $x(0)$ rests at its stationary distribution or, equivalently, the following equality holds:
\begin{align}\nonumber
x(0)=\lim_{\tilde{T}\rightarrow\infty}\sum_{\tau=-\tilde{T}}^{-1}(A+BK_0)^{-\tau-1}(w(\tau)+Bv(\tau))
\end{align}
Note that, for exponentially stable systems, the state converges to its stationary distribution exponentially fast and therefore, the stationarity of $x(0)$ is a reasonable assumption. Furthermore, using the above equality, it is easy to see that $x(0)$ is a random vector whose elements are (dependent) centered sub-Gaussian random variables with bounded parameters. Moreover, one can verify that its covariance $\mathbb{E}\{x(0)x(0)^\top\} = Q^*$ satisfies the following Lyapunov equation:
\begin{align}
(A+BK_0)Q^*(A+BK_0)^\top-Q^*+\sigma^2_wI + \sigma_v^2BB^\top = 0
\end{align}
\begin{sloppypar}
	Accordingly, $Q^*$ can be used to derive the covariance matrix $M^*$ for the random vector $\begin{bmatrix}
	x(0)^\top & (K_0x(0)+v(0))^\top
	\end{bmatrix}^\top$:
	\begin{equation}\nonumber
	{M}^* = \begin{bmatrix}
	Q^* & Q^*K_0^T\\
	K_0Q^* & K_0Q^*K_0^T + \sigma^2_v I 
	\end{bmatrix}
	\end{equation}
	Define $\mathcal{A}_j = \{i: \Psi^*_{ij} \not = 0\}$ and let $\mathcal{A}^c_j$ refer to its complement. Denote $k$ as the maximum number of nonzero elements in any column of ${\Psi^*}$. 
\end{sloppypar}

\begin{assumption}
	The following inequalities are satisfied
	\begin{itemize}
		\item[A1] (Mutual incoherence)
		\begin{align}\nonumber
		\max_{1\leq j\leq n}\left\{\max_{i\in\mathcal{A}^c_j}\left\{\left\|{{M}^*_{i \mathcal{A}_j}(M^*_{\mathcal{A}_j \mathcal{A}_j})^{-1}}\right\|_{1}\right\}\right\}\leq 1-\gamma
		\end{align}
		\item[A2] (Bounded eigenvalue)
		\begin{align}\nonumber
		\min_{1\leq j\leq n}\lambda_{\min}(M^*_{\mathcal{A}_j \mathcal{A}_j})\geq C_{\min}
		\end{align}
		\item[A3] (Bounded infinity norm)
		\begin{align}\nonumber
		\max_{1\leq j\leq n}\vertiii{(M^*_{\mathcal{A}_j \mathcal{A}_j})^{-1}}_\infty\leq D_{\max}
		\end{align}
		\item[A4] (Nonzero gap)
		\begin{align}\nonumber
		\min_{1\leq j\leq n}\left\{\max_{i\in\mathcal{A}_j}\left\{|\Psi^*_{ij}|\right\}\right\} \geq \Psi_{\min}
		\end{align}
	\end{itemize}
	for some constants $0<\gamma<1$, $1\geq C_{\min}>0$, $D_{\max}\geq1$ and $1\geq \Psi_{\min}>0$.
\end{assumption}
Next, we present the main result of the paper.
\begin{theorem}\label{thm1_single}
	Assume that $k\geq 2$ and
	\begin{align}
	&\lambda = c_1\cdot\frac{C}{1-\rho}\cdot\frac{\eta^2}{\gamma}\sqrt{\frac{\log((n+m)/\delta)}{T}}\label{lower_lambda}\\
	& T\geq c_2\cdot\frac{C^4}{(1-\rho)^4}\cdot\frac{D^2_{\max}}{\gamma^2C^2_{\min}\Psi^2_{\min}}\cdot k^2\log((n+m)/\delta),\label{lower_T}
	\end{align}
	where $c_1$ and $c_2$ are universal constants. Then, the following statements hold with probability of at least $1-\delta$:
	\begin{itemize}
		\item[1.](Correct sparsity recovery)~\eqref{ls2_single} has a unique solution and recovers the true sparsity pattern of $\Psi^*$.
		\item[2.] ($\ell_{\infty}$-norm error) We have
		\begin{equation}\label{est_err}
		\|\hat\Psi-\Psi^*\|_{\infty}\leq c_3\cdot\frac{C}{1-\rho}\cdot\frac{D_{\max}\eta^2}{\gamma}\sqrt{\frac{\log((n+m)/\delta)}{T}}
		\end{equation}
		where $c_3$ is a universal constant.
	\end{itemize}
\end{theorem}

\begin{remark}
	As mentioned before, the injection of a random input noise is essential to guarantee the identifiability of the parameters. This is also reflected in the above theorem: in order to guarantee a finite sample complexity for the proposed estimator, it is crucial to have $C_{\min}>0$, which is only possible if $\sigma_v>0$.
\end{remark}

A number of observations can be made based on Theorem~\ref{thm1_single}. First, it implies that if $\gamma$, $C$, $D_{\max}$, $C_{\min}$, $\Psi_{\min}$, and $\rho$ do not scale with the system dimension, then $T = \Omega(k^2\log(n+m))$ is enough to guarantee the correct sparsity recovery and a small estimation error. Notice that for sparse systems, this quantity can be much smaller than the system dimension. Second, the sample complexity of the proposed estimator depends on $\frac{C}{1-\rho}$, which is a measure of the system stability. In particular, for highly stable systems, $\frac{C}{1-\rho}$ is small, resulting in an improved accuracy of the proposed estimator with smaller $T$. In contrast, when the system is close to its stability margin, $\frac{C}{1-\rho}$ will grow which negatively affects the estimation error as well as the lower bound on $T$. Another intuitive interpretation of $\frac{C}{1-\rho}$ is that it measures the amount of \textit{dependency} between the states at different times: for highly stable systems where $\rho$ is small, $(x(t),u(t))$ is only weakly dependent on $(x(\tau),u(\tau))$ for $\tau = 0,\dots,t-1$, thereby facilitating the estimation of the unknown parameters. 
We finally mention that this dependency is in contrast with the recent discoveries on the sample complexity of the least-squares estimator, which support the favorable effect of a large $\rho$ on the accuracy of the estimator~\cite{simchowitz2018learning}. We leave investigating whether this seemingly contradictory observation is an artifact of our methodology (e.g., mixing the initial state to the stationary distribution), or is fundamental to the sparse system identification problem, to future work.
\begin{remark}
	In order to further enhance the accuracy of the proposed estimator, one can perform a least-squares estimation restricted to the nonzero elements of the estimated parameter, after obtaining its sparsity pattern via the proposed method. Although, theoretically, this post-model-selection estimation method may not improve the estimation error rate, it will incur less bias~\cite{belloni2013least}. We will show in our simulations that the effect of this post-processing step can be significant in the accuracy of the estimation.
\end{remark}
\subsection{Comparison to prior art}

As mentioned before, another line of work focuses on unstructured system identification, where either the learning time $T$ or the number of sample trajectories $d$ is allowed to grow. In~\cite{dean2017sample}, the authors consider the sample complexity of the system identification problem with multiple sample trajectories via least-squares, where it is shown that the proposed estimator incurs a small error, provided that $d = \Omega(n+m)$. Revisiting~\eqref{lower_T} reveals that the proposed method outperforms the sample complexity of ordinary least-squares when $k$ is significantly smaller than $n+m$, i.e., exploiting prior knowledge of the system sparsity leads to a reduction in sample complexity. In~\cite{sarkar2018fast, simchowitz2018learning, abbasi2011regret, faradonbeh2018finite}, the authors consider unstructured system identification from a single sample trajectory under different assumptions on system stability and/or the initial state of the system. However, similar to~\cite{dean2017sample}, none of these works take advantage of the underlying sparsity structures of the system matrices. As a result, they cannot correctly estimate the sparsity structure of $(A,B)$ and suffer from poor dependencies on the system dimensions in the large-scale and structure setting.

Subsequently, a Lasso-type estimator is proposed in~\cite{Salar18} to further exploit the underlying sparsity pattern of $(A,B)$ with $d$ sample trajectories, each with a zero initial state. In particular, it is shown that $d = \Omega\left(\frac{\kappa({\Sigma})^2}{\gamma^2\Psi_{\min}^2}k\log(n+m)\right)$ is enough to ensure the correct sparsity recovery and a small estimation error with high probability, where $\kappa({\Sigma})$ is the condition number of the finite-time \textit{controllability matrix} of the system. Comparing this quantity with~\eqref{lower_T}, one can observe that the former has a better dependency on $k$. However, $\kappa({\Sigma})$ is highly dependent on the learning time $T$. In fact, it is easy to show that for unstable systems, $\kappa({\Sigma})$ may grow exponentially fast with respect to $T$. On the other hand,~\eqref{lower_T} is free of such dependency and instead, it is in terms of the stationary distributions of the state and input vectors.

Moreover, our work is a major extension to the results of~\cite{pereira2010learning}, where the authors address a similar sparse system identification problem with a single sample trajectory. First, unlike the presented results,~\cite{pereira2010learning} only considers autonomous systems, i.e., systems~\eqref{ss} with $B$=0. Second,~\cite{pereira2010learning} only ensures the correct sparsity recovery of the true parameters. In contrast, we extend these results to obtain non-asymptotic bounds on the estimation error. As demonstrated in~\cite{dean2017sample, dean2018regret}, having these bounds is essential for the design of near-optimal and robustly stabilizing controllers. Third,~\cite{pereira2010learning} requires that the closed-loop system be contractive with respect to the spectral norm, i.e., that $|||(A+BK_0)|||<1$, whereas we only require system stability. Notice that the former condition is much stronger, as in practice, stable systems are often not contractive in spectral norm. Finally, the validity of the non-asymptotic bounds introduced in~\cite{pereira2010learning} heavily relies on the Gaussian nature of the disturbance and input noises. As an extension to this result, our proposed method targets a larger class of uncertainties for the disturbance and input noises, thereby allowing for norm bounded disturbance and input signals.

\subsection{Mutual incoherency}

In this subsection, we analyze the mutual incoherence condition on the steady-state covariance matrix $M^*$. In particular, we explain why this assumption is not an artifact of the proposed method, but that it rather stems from a fundamental limitation of \textit{any} sparsity-promoting technique for the system identification problem. We show that similar mutual incoherence assumptions are indeed necessary to recover the correct sparsity of system parameters by using a class of \textit{oracle estimators}.

We assume that the oracle estimator can measure the disturbance matrix $W$ and that it can work with sample trajectories of an arbitrary length.
With these assumptions, the oracle estimator solves the following optimization problem to estimate the parameters of the system:
\begin{subequations}\label{opt}
	\begin{align}
	\min_{\Psi}\ \ & \|\Psi\|_0\\
	\mathrm{s.t.}\ \ & X\Psi = Y-W\label{opt2}
	\end{align}
\end{subequations}
Clearly, this oracle estimator cannot be used in practice since 1) the disturbance matrix $W$ is unknown, 2) the learning time $T$ is finite, and 3) the corresponding optimization problem is non-convex and NP-hard in its worst case. Setting aside these restrictions for now, there are fundamental limits on the consistency of this estimator. To explain this, we introduce the mutual-coherence metric for a matrix (note the difference between this definition and Assumption A1). For a given matrix $A\in\mathbb{R}^{t_1\times t_2}$, its mutual-coherence $\mu(A)$ is defined as
\begin{align}\nonumber
\mu(A) = \max_{1\leq i<j\leq t_2}\frac{|A_{:,i}^\top A_{:,j}|}{\|A_{:,i}\|_2\|A_{:,j}\|_2}
\end{align}
In other words, $\mu(A)$ measures the maximum correlation between distinct columns of $A$. Reminiscent of the classical results in the compressive sensing literature, it is well-known that the optimal solution $\Psi^*$ of~\eqref{opt} is unique if the following \textit{identifiability} condition
\begin{align}\label{cond1}
\|\Psi^*_{:,j}\|_0<\frac{1}{2}\left(1+\frac{1}{\mu(X)}\right)
\end{align} 
holds for $j = 1,2,...,n$ (see, e.g., Theorem 2.5 in~\cite{elad2010sparse}). Furthermore, this bound is tight, implying that there exists an instance of the problem for which the violation of $\|\Psi^*_{:,j}\|_0<\frac{1}{2}\left(1+\frac{1}{\mu(X)}\right)$ for some $j$ results in the non-uniqueness of the optimal solution. On the other hand, according to Lemma~\ref{l2_app} (to be introduced later) and the Borel-Cantelli lemma, $\frac{1}{T}X^\top X$ converges to $M^*$ almost surely, as $T\rightarrow \infty$. This implies that
\begin{align}\nonumber
\mu(X) &= \max_{1\leq i<j\leq m+n}\frac{|X_{:,i}^\top X_{:,j}|}{\|X_{:,i}\|_2\|X_{:,j}\|_2} \overset{a.s.}\rightarrow \max_{1\leq i<j\leq m+n}\frac{|M^*_{ij}|}{\sqrt{M^*_{ii}M^*_{jj}}}
\end{align}
The above analysis reveals that the off-diagonal entries of $M^*$ play a crucial role in the identifiability of the true parameters: as these elements become smaller relative to the diagonal entries, the oracle estimator can correctly identify the structure of $\Psi$ for a wider range of sparsity levels. 
Similarly, our proposed mutual incoherence assumption is expected to be satisfied when the off-diagonals of $M^*$ have small magnitudes, relative to the diagonal entries. This implies that Assumption A1 is a natural condition to impose in order to ensure the correct sparsity recovery of $\Psi$. Furthermore, in practice, $M^*$ will be close to a diagonally dominant matrix with exponentially decaying off-diagonal entries, provided that the matrices $A$, $B$, and $K_0$ have sparse structures~\cite{simoncini2015lyapunov}. 

\section{Numerical Experiments}
\label{sec:example}

\begin{figure*}
	\centering
	\subfloat[Relative mismatch error]{\label{fig_ME}
		\includegraphics[width=.33\columnwidth]{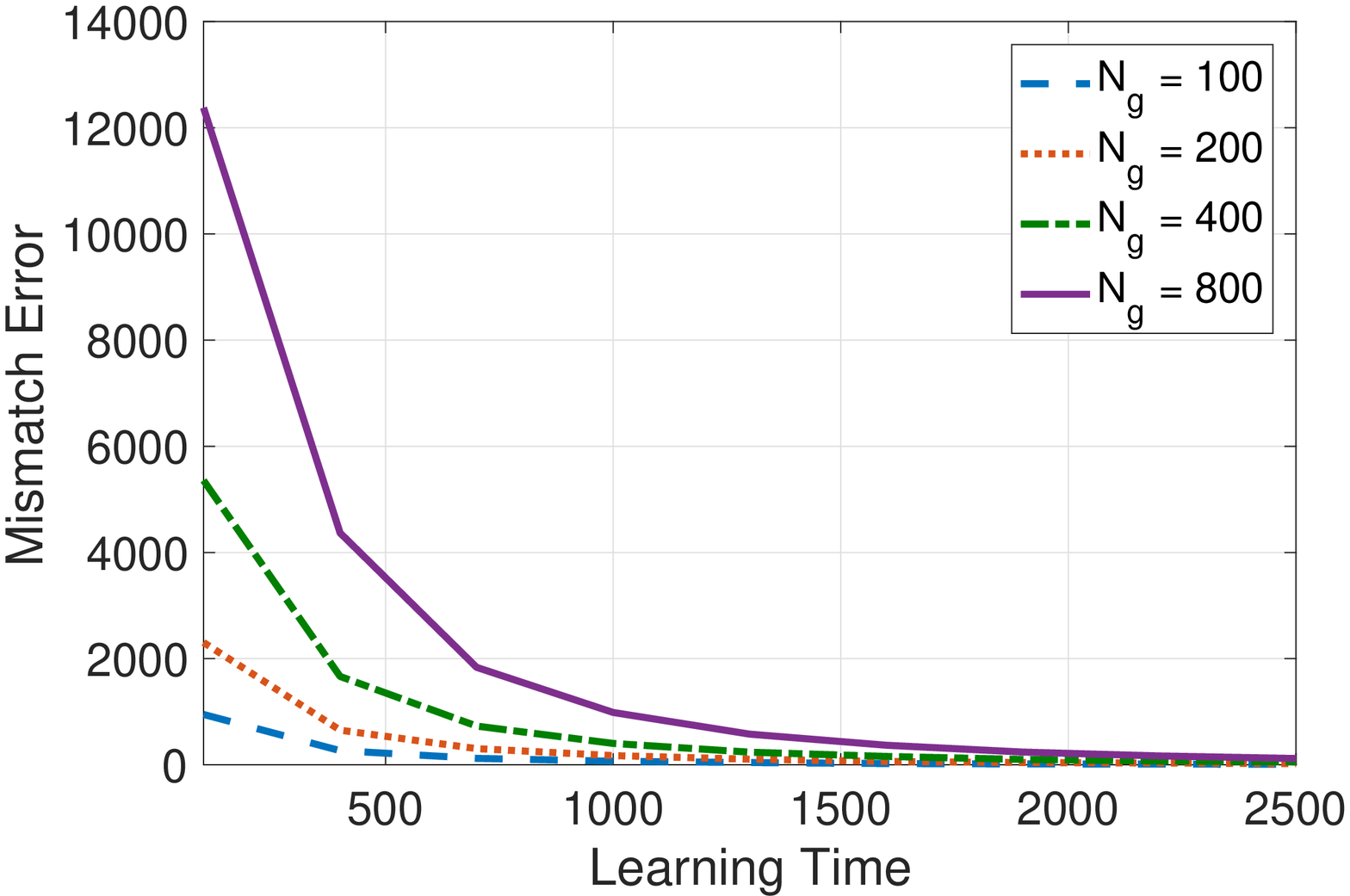}}
	\subfloat[Normalized estimation error]{\label{fig_est_error}
		\includegraphics[width=.33\columnwidth]{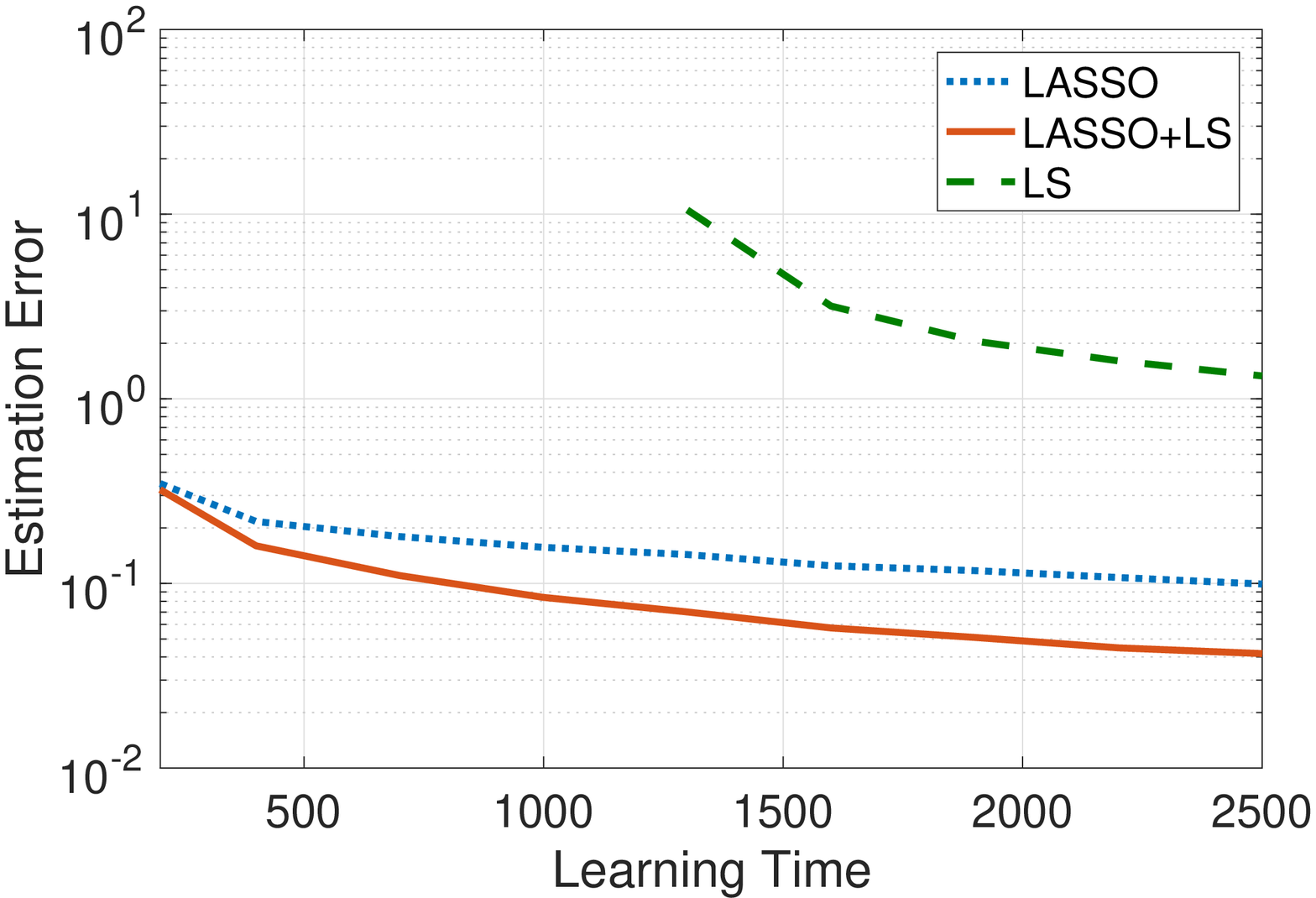}}
	\subfloat[The distribution of $\gamma$]{\label{fig_gamma}
		\includegraphics[width=.33\columnwidth]{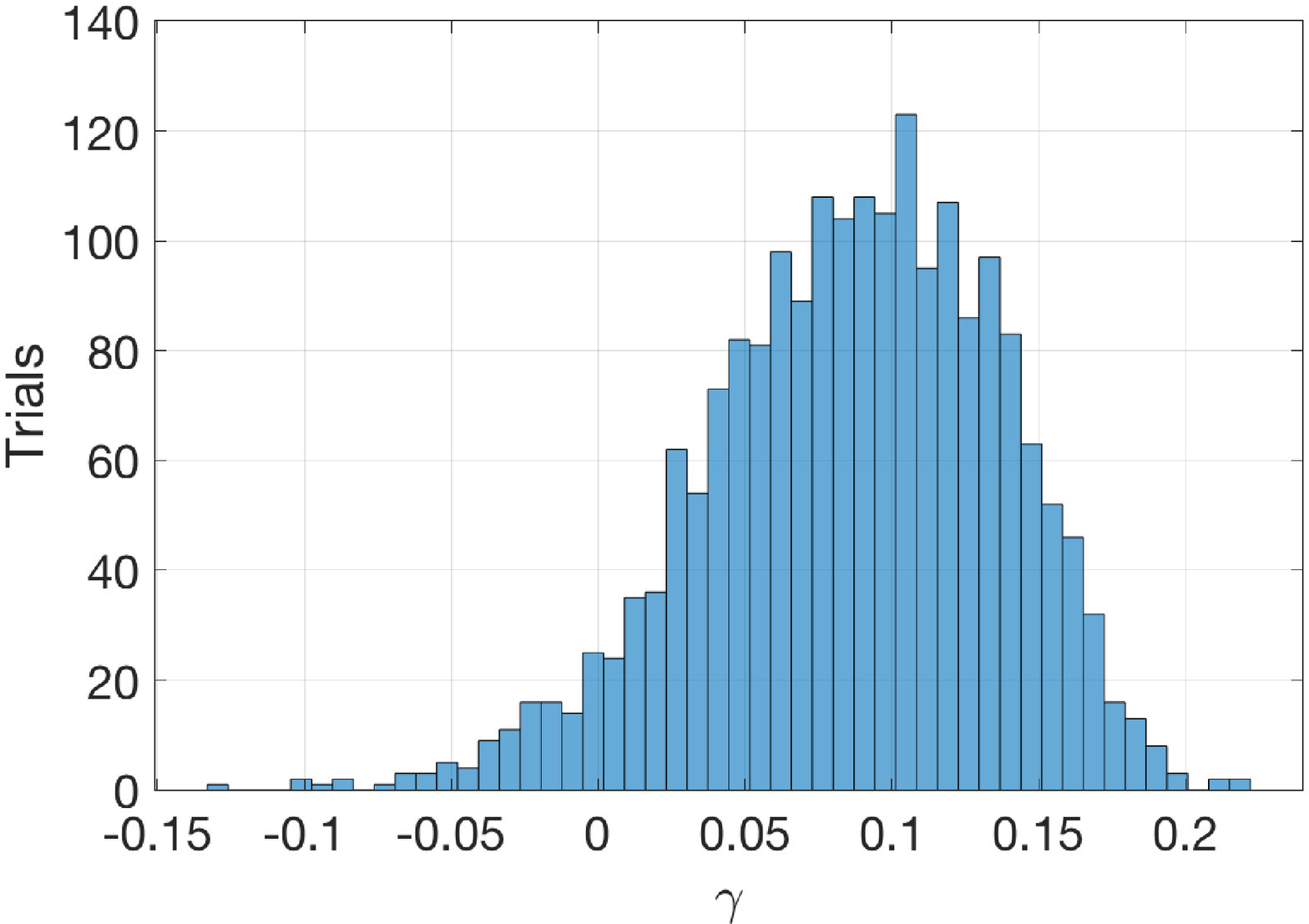}}
	\caption{ \footnotesize (a) The mismatch error with respect to the learning time for different number of generators in the system. The values are averaged over 10 independent trials. (b) The normalized estimation error for Lasso (abbreviated as \texttt{LASSO}), Lasso + least-squares (abbreviated as \texttt{LASSO+LS}), and least-squares (abbreviated as \texttt{LS}) estimators with respect to the learning time. The values are averaged over 10 independent trials. (c) The distribution of mutual incoherence parameter $\gamma$ for 2000 randomly generated instances of the problem.}
\end{figure*}

As a case study, we consider the frequency control problem for power systems, where the goal is to control the governing frequency of the entire network, based on the so-called \textit{swing} equations. Assume that there exist $N_g$ generators in the system. It is common to describe the per-unit swing equations using the well-known direct current (DC) approximation:
\begin{align}\nonumber
M_i\ddot{\theta}_i+D_i\dot{\theta}_i = P_{M_i} - P_{E_i}
\end{align}
where $\theta_i$ is the voltage angle at generator $i$, $P_{M_i}$ is the mechanical power input at generator $i$, and $P_{E_i}$ denotes the active power injection at the bus connected to generator $i$. Furthermore, $M_i$ and $D_i$ are the inertia and damping coefficients at generator $i$, respectively. Under the DC approximation, the relationship between active power injection and voltage is defined as follows:
\begin{align}\nonumber
P_{E_i} = \sum_{j\in\mathcal{N}_i}B_{ij}(\theta_i-\theta_j)
\end{align}
where $n$ is the number of generators in the network, $\mathcal{N}_i$ collects the neighbors of generator $i$, and $B_{ij}$ is the susceptance of the line $(i,j)$. After discretization with the sampling time $dt$, the system of swing equations is reduced to the following dynamical system:
\begin{align}\nonumber
x_{i}(t+1) = \left(A_{ii}x_i(t)+\sum_{j\in\mathcal{N}_i}A_{ij}x_j(t)\right)+B_{ii}u_i(t)+w_i(t)
\end{align}
where $x_{i} = \begin{bmatrix}
\theta_i & \dot{\theta_i}
\end{bmatrix}^\top$, $u_i(t) = P_{M_i}$, and 
\begin{align}\nonumber
A_{ii} \!=\! \begin{bmatrix}
1 & dt\\
-\frac{\sum_{j\in\mathcal{N}_i}B_{ij}}{M_i}dt & 1-\frac{D_i}{M_i}dt
\end{bmatrix}, A_{ij} \!=\! \begin{bmatrix}
0 & 0\\
\frac{B_{ij}}{M_i}dt & 0
\end{bmatrix}, B_{ii} \!=\! \begin{bmatrix}
0\\
1
\end{bmatrix}
\end{align}
The goal is to identify the underlying dynamical system based on a single sample trajectory consisting of a sequence of mechanical power inputs and their effects on the angles and frequencies of different generators. To assess the performance of the proposed method, we generate several instances of the problem according to the following rules:
\begin{itemize}
	\item[-] the generators are connected via a randomly generated tree with a maximum degree of $10$.
	\item[-] the parameters $B_{ij}$, $M_i$, $D_i$ are uniformly chosen from $[0.5,1]$, $[1,2]$, $[0.5,1.5]$, respectively.
\end{itemize}
Furthermore, the sampling time $dt$ is set to $0.1$. We assume that the disturbance noise has a zero-mean Gaussian distribution with covariance $0.01I_{2\times 2}$. Notice that the magnitude of the noise is comparable to those of the nonzero elements in $A$ and $B$. Furthermore, the mechanical input is set to $u_i(t) = -0.1(\theta_i+\dot{\theta_i})+v_i(t)$, where $v_i(t)$ is a randomly generated input noise, distributed according to a zero-mean Gaussian distribution with variance $0.05$. Notice that the first term in the input signal is used to ensure the closed-loop stability. 

The reported results are for a serial implementation in MATLAB R2017b, and the function \texttt{lasso} is used to solve~\eqref{ls_single}. It is worthwhile to note that the running time can be further reduced via parallelization; this is trivially possible due to the decomposable nature of the problem. The \textit{mismatch error} is defined as the total number of false positives and false negatives in the sparsity pattern of the estimated parameters $(\hat{A},\hat{B})$. Furthermore, \textit{relative learning time} (RLT) is defined as the learning time normalized by the dimension of the system, and \textit{relative mismatch error} (RME) is used to denote the mismatch error normalized by the total number of elements in $A$ and $B$. In all of our experiments, the regularization coefficient $\lambda$ is set to $\lambda = \sqrt{\frac{0.03\log(n+m)}{T}}$. Note that this value does not require any additional fine-tuning and is at most a constant factor away from~\eqref{lower_lambda}.

Figure~\ref{fig_ME} illustrates the mismatch error (averaged over 10 different trials) with respect to the learning time $T$ and for different number of generators $N_g$ that are chosen from $\{100,200,400,800\}$. These correspond to the total system dimensions of $\{300, 600, 1200, 2400\}$. Note that the largest instance has more than $3.84$ million unknown parameters. Not surprisingly, the learning time needed to achieve a small mismatch error increases as the dimension of the system grows. Conversely, a smaller value for RLT is needed to achieve infinitesimal RME for larger systems. In particular, when $N_g$ is equal to $100$, $200$, $400$, and $800$, the minimum RLT to guarantee RME $\leq 0.1\%$ is equal to $3.83$, $1.42$, $0.50$, and $0.16$, respectively.

As mentioned before, the accuracy of the proposed estimator can be improved by additionally applying the least-squares over the nonzero elements of $(\hat{A}, \hat{B})$. Figure~\ref{fig_est_error} illustrates the normalized 2-norm estimation error of this approach (abbreviated as \texttt{LASSO+LS}), compared to the proposed method without any post-processing step (abbreviated as \texttt{LASSO}), and the least-squares estimator (abbreviated as \texttt{LS}) when $N_g$ is set to $200$. It can be observed that both \texttt{LASSO+LS} and \texttt{LS} significantly outperform \texttt{LS}; in fact, \texttt{LS} is not even well-defined if the learning time is strictly less than the system dimensions. Furthermore, on average, the estimation error for \texttt{LASSO+LS} is $1.91$ times smaller than that of \texttt{LASSO}.

Finally, only 32 out of 360 generated instances did not satisfy the proposed mutual incoherence condition. However, this violation did not have a significant effect on the accuracy of the proposed estimator. To further investigate the frequency of the instances that satisfy this condition, we plot the histogram of the mutual incoherence parameter $\gamma$ for 2000 randomly generated instances with fixed $N_g = 200$. It can be seen in Figure~\ref{fig_gamma} that the mutual incoherence condition is violated only for $5.15\%$ of the instances.

\section{Conclusions}\label{sec:conclusion}
The problem of sparse system identification of linear time-invariant (LTI) systems is considered in this work, where the goal is to estimate the sparse structure of the system matrices based on a single sample trajectory of the dynamics. A Lasso-type estimator is introduced to identify the parameters of the system, while promoting their sparsity via a $\ell_1$-regularization technique. By carefully examining the underlying properties of the system---such as its stability and mutual incoherency---we provide non-asymptotic bounds on the accuracy of the proposed estimator. In particular, we show that it correctly identifies the sparsity structure of the system matrices and enjoys a sharp upper bound on its estimation error, provided that the learning time exceeds a threshold. We further show that this threshold scales polynomially in the number of nonzero elements but logarithmically in the system dimensions.

\bibliographystyle{IEEEtran}
\bibliography{bib.bib}

\appendix
\section{Proof of Theorem~\ref{thm1_single}}
In this section, we present the sketch of the proof for the main theorem. 
Define
\begin{equation}\nonumber
L(\Psi_{:,j}) = \|Y-X\Psi_{:,j}\|^2_2
\end{equation}
and 
\begin{equation}\label{lsj}
\hat{\Psi}_{:,j} = \arg\min \frac{1}{2T}L(\Psi_{:,j})+\lambda\|\Psi_{:,j}\|_1
\end{equation}
for every $j\in\{1,2,..., n\}$. It is easy to verify that 
\begin{equation}\nonumber
\hat\Psi = \begin{bmatrix}
\hat{\Psi}_{:,1} & \hat{\Psi}_{:,2} & \cdots & \hat{\Psi}_{:,n}
\end{bmatrix}
\end{equation}
Furthermore, the Gradient and Hessian of $L(\cdot)$ are equal to
\begin{align}
& G = -\nabla L(\Psi_{:,j})|_{\Psi_{:,j} = \Psi_{:,j}^*} = \frac{1}{T}X^TW_{:,j}, \nonumber\\
& M = \nabla^2 L(\Psi_{:,j})|_{\Psi_{:,j} = \Psi^*_{:,j}} = \frac{1}{T} X^TX\nonumber
\end{align}
Note that $G$ can be different for every $j$. However, we keep this dependency implicit in the notations to streamline the presentation.
The following Lemma is at the core of our subsequent analysis:
\begin{lemma}[Proposition 4.1 \cite{pereira2010learning}]\label{l_sparse}
	Suppose that the following conditions are satisfied:
	\begin{align}
	&\|G\|_{\infty}\leq\frac{\lambda\gamma}{3},\nonumber\\
	&\|G_{\mathcal{A}_j}\|_{\infty}\leq \frac{\Psi_{\min}C_{\min}}{4k}-\lambda\nonumber\\
	& \vertiii{M_{\mAj^c \mAj}\!-\!M^*_{\mAj^c \mAj}}_{\infty}\!\!\leq \frac{\gamma C_{\min}}{12\sqrt{k}},\nonumber\\
	&\vertiii{M_{\mAj \mAj}\!-\!M^*_{\mAj \mAj}}_{\infty}\!\!\leq \frac{\gamma C_{\min}}{12\sqrt{k}}\nonumber
	\end{align}
	Then,~\eqref{lsj} recovers the true sparsity pattern of $\Psi^*_{:,j}$.
\end{lemma}

The first step in proving Theorem~\ref{thm1_single} is to verify that the conditions of Lemma~\ref{l_sparse} hold with high probability. To this goal, first we write $x(t)$ and $u(t)$ in terms of $x(0)$, $w(\tau)$ and $v(\tau)$ for $\tau = 0,1,\dots, t$:
\begin{align}
x(t) =& (A+BK_0)^tx(0)+\sum_{\tau=0}^{t-1}(A+BK_0)^{t-\tau-1}(w(\tau)+Bv(\tau))\nonumber\\
u(t) =& v(t)+K_0(A+BK_0)^tx(0)+\sum_{\tau=0}^{t-1}K_0(A+BK_0)^{t-\tau-1}(w(\tau)+Bv(\tau))\nonumber
\end{align}
Instead of initiating the system at $x(0)$ with the stationary distribution, we will start at the time $-T_0$, with a modified initial state $x(-T_0) = w(-T_0-1)+Bv(-T_0-1)$, where $w(-T_0-1)$ and $v(-T_0-1)$ have the same distributions as the disturbance and input noises, respectively. Since the system is stable, by taking $T_0\rightarrow\infty$ and invoking the Continuous Mapping Theorem, the matrices
\begin{align}\nonumber
\begin{bmatrix}
x(0) & x(1) &\dots & x(T-1)
\end{bmatrix}
\end{align}
and 
\begin{align}
\begin{bmatrix}\nonumber
K_0x(0)\!+\!v(0)\! &\! K_0x(1)\!+\!v(1) \!&\!\dots \!&\! K_0x(T\!-\!1)\!+\!v(T\!-\!1)
\end{bmatrix}
\end{align}
converge in distribution to the same matrices when the system is initialized at a state with the stationary distribution. Therefore, without loss of generality, we will focus on the former. Based on this observation, one can write
\begin{align}
& x(t) = \lim\limits_{T_0\rightarrow\infty}\sum_{\tau=-T_0-1}^{t-1}(A+BK_0)^{t-\tau-1}(w(\tau)+Bv(\tau))\nonumber\\
& u(t) = v(t)\!+\!\lim\limits_{T_0\rightarrow\infty}\!\sum_{\tau=-T_0-1}^{t-1}\!K_0(A+BK_0)^{t-\tau-1}(w(\tau)\!+\!Bv(\tau))\nonumber
\end{align}
This implies that the elements in $G$ and $M$ can be written as quadratic functions of the disturbance and input noises in the form of $G_i = z^\top R_G z$ and $M_{ij} = z^\top R_M z$, where $z\in\mathbb{R}^{(n+m)(t+T_0+1)}$ is a random vector, defined as
\begin{align}\nonumber
z \!=\! \begin{bmatrix}
w(-T_0\!-\!1)^\top \!&\! \cdots \!&\! w(t-1)^\top \!&\! v(-T_0\!-\!1)^\top &\!\cdots\! \!&\! v(t-1)^\top
\end{bmatrix}^\top
\end{align}
The following theorem will be used in our analysis to provide concentration bounds on $G$ and $M$.
\begin{theorem}[Hanson-Wright inequality~\cite{rudelson2013hanson}]
	Let $x = \begin{bmatrix}
	x_1 & x_2 & \dots & x_n
	\end{bmatrix}$ be a random vector with independent zero-mean sub-Gaussian elements. Given a square and symmetric matrix $P$, the following inequality holds
	\begin{align}\nonumber
	\mathbb{P}\left(\left|x^\top Px - \mathbb{E}\left\{x^\top Px\right\}\right|>t\right)
	\leq 2\exp\left(-c\cdot\min\left\{\frac{t^2}{\|x\|^4_{\psi}\|P\|^2_F}, \frac{t}{\|x\|^2_{\psi}\vertiii{P}}\right\}\right)
	\end{align}
	for every $t\geq 0$, where $c$ is a universal constant.
\end{theorem}
For a symmetric matrix $P$, we have $\|P\|^2_F = \sum_{k=1}^{n}\lambda^2_k$. Therefore, the above theorem implies that, for a sub-Gaussian random vector $z$ with independent elements, we have
\begin{align}
\mathbb{P}\left(\left|z^\top Pz \!-\! \mathbb{E}\left\{z^\top Pz\right\}\right|>t\right)\leq 2\exp\left(-c\cdot\frac{t^2}{\|z\|^4_{\psi}\left(\sum_{k=1}^{n}\lambda^2_k\right)}\right)\nonumber
\end{align}
provided that $t\leq \left(\frac{\sum_k\lambda_k^2}{\max_k|\lambda_k|}\right)\|z\|^2_\psi$. The assumptions of Lemma~\ref{l_sparse} can be seen to hold directly as a consequence of the following two lemmas:
\begin{lemma}\label{l1_app}
	Let $i\in\{1,2,..., n+m\}$ and suppose that $\epsilon<\frac{3C\eta^2}{1-\rho}$. Then, there exists a universal constant $c_4$ such that
	\begin{equation}\nonumber
	\mathbb{P}\{|G_i|>\epsilon\}\leq 2\exp\left(-c_4\frac{(1-\rho)^2}{C^2\eta^4}T\epsilon^2\right)
	\end{equation}
\end{lemma}
\begin{proof}
	See Appendix~\ref{pr_l1_app}.
\end{proof}
\begin{lemma}\label{l2_app}
	Let $i,j\in\{1,2,...,n+m\}$ and suppose that $\epsilon\leq\frac{4C^2\eta^2}{(1-\rho)^2}$. Then, there exists a universal constant ${c}_5$ such that
	\begin{equation}\nonumber
	\mathbb{P}\{|M_{ij}-M^*_{ij}|>\epsilon\}\leq 2\exp\left(-{c}_5\frac{(1-\rho)^4}{C^4\eta^4}T\epsilon^2\right)
	\end{equation}
\end{lemma}
\begin{proof}
	See Appendix~\ref{pr_l2_app}.
\end{proof}

The following proposition shows that for a fixed column $j$, the proposed estimator~\eqref{lsj} correctly recovers the sparsity pattern with high probability.
\begin{proposition}\label{prop_recovery}
	Assume that $k\geq 2$ and the following conditions are satisfied:
	\begin{align}
	&\lambda = c_6\cdot\sqrt{\frac{C^2\eta^4}{\gamma^2T(1-\rho)}\log(n+m/\delta)}\label{lambda}\\
	& T\geq c_7\cdot\frac{C^4\eta^4k^2}{\gamma^2C^2_{\min}\Psi^2_{\min}(1-\rho)^4}\log(n+m/\delta)\label{lowerT}
	\end{align}
	for universal constants $c_6, c_7\geq 0$.
	Then,~\eqref{lsj} recovers the true sparsity pattern of $\Psi^*_{:,j}$ with probability of at least $1-\delta$.
\end{proposition}
\begin{proof}
	The Lemmas~\ref{l1_app} and~\ref{l2_app} can be used to prove statement. The details are provided in Appendix~\ref{pr_correct}.
\end{proof}

The next lemma provides a deterministic upper bound on the estimation error in terms of the deviations of $M$ and $G$ from their mean.
\begin{lemma}\label{detbound_app}
	Assume that 
	\begin{equation}\label{cmin2}
	\vertiii{M_{\mAj, \mAj}-M^*_{\mAj, \mAj}}_{\infty}\leq \frac{\min\{1,2\eta^2\}}{2D_{\max}}
	\end{equation}
	and~\eqref{lsj} recovers the correct sparsity pattern of $\Psi^*_{:,j}$. Then, the following inequality holds for $E = \hat{\Psi}_{:,j}-\Psi^*_{:,j}$:
	\begin{align}\label{err}
	& E_{\mAj^c} = 0\nonumber\\
	&\|E_{\mAj}\|_{\infty}\!\leq\! \left(2D^2_{\max}\vertiii{M_{\mAj \mAj}\!\!-\!M^*_{\mAj \mAj}}_{\infty}\!\!+\!D_{\max}\right)\left(\|G_{\mAj}\|_{\infty}\!\!+\!\lambda\right)
	\end{align}
\end{lemma}
\begin{proof}
	See Appendix~\ref{pr_detbound_app}.
\end{proof}

The next lemma shows that the condition of Proposition~\ref{detbound_app} holds with high probability, provided that $T$ is large enough.
\begin{proposition}\label{prop_inf_bound}
	Assume that
	\begin{equation}\label{T_lower}
	T\geq c_8\cdot\frac{D^2_{\max}C^4}{(1-\rho)^4}k^2\log(k/\delta)
	\end{equation}
	for some universal constant $c_5\geq 0$. Then, the following inequality holds with probability of at least $1-\delta$
	\begin{equation}
	\vertiii{M_{\mAj, \mAj}-M^*_{\mAj, \mAj}}_\infty\leq \frac{\min\{1,2\eta^2\}}{2D_{\max}}
	\end{equation}
\end{proposition}
\begin{proof}
	Notice that $|\mAj|\leq k$. One can verify that 
	\begin{equation}\label{eq114}
	\mathbb{P}\left(\vertiii{M_{\mAj, \mAj}-M^*_{\mAj, \mAj}}_\infty>\epsilon\right)\leq 2k^2\exp\left(-c_5\cdot\frac{(1-\rho)^4}{C^4\eta^4}\frac{T}{k^2}\epsilon^2\right)
	\end{equation}
	provided that $\frac{\epsilon}{k}\leq\frac{4C^2\eta^2}{(1-\rho)^2}$. Setting $\epsilon = \frac{\min\{1,2\eta^2\}}{2D_{\max}}$ and recalling that $D_{\max}, C\geq 1$, one can verify that $\frac{\epsilon}{k}\leq \frac{4C^2\eta^2}{(1-\rho)^2}$ is satisfied. Furthermore, by choosing $c_8 = \frac{16}{c_5}$, one can certify that~\eqref{T_lower} is enough to ensure that the right hand side of the above inequality is upper bounded by $\delta$, thereby completing the proof.
\end{proof}

\noindent{\it Proof of Theorem~\ref{thm1_single}:}
First note that~\eqref{ls2_single} can be decomposed into $n$ disjoint sub-problems over different columns of $\Psi$, each in the form of~\eqref{lsj}.
Consider the following choices for $\lambda$ and $T$:
\begin{align}
&\lambda = c_6\cdot\sqrt{\frac{C^2\eta^4}{\gamma^2T(1-\rho)^2}\log(4(n+m)/\delta)}\\
& T\geq \max\left\{c_7,c_8,\frac{1}{c_4},\frac{2}{c_5}\right\}\cdot\frac{C^4D^2_{\max}k^2}{\gamma^2C^2_{\min}\Psi^2_{\min}(1-\rho)^4}\log((n+m)/\delta)\label{lower_T}
\end{align}
where $c_4$, $c_5$, $c_6$, $c_7$, and $c_6$ are introduced in Lemmas~\ref{l1_app},~\ref{l2_app}, and Propositions~\ref{prop_recovery},~\ref{prop_inf_bound}. Based  on the Proposition~\ref{prop_recovery} and the above choices for $\lambda$ and $T$,~\eqref{lsj} recovers the sparsity pattern of $\Psi^*_{:,j}$ for a given column index $j$ with probability of at least $1-\delta$. Furthermore, based on Proposition~\ref{prop_inf_bound}, the lower bound on $T$ guarantees that the inequality
\begin{equation}\label{cmin}
\vertiii{Q_{\mAj, \mAj}-Q^*_{\mAj, \mAj}}_{\infty}\leq \frac{\min\{1,2\eta^2\}}{2D_{\max}}
\end{equation}
holds with probability of at least $1-\delta$. This, together with Proposition~\ref{detbound_app} results in
\begin{equation}\label{eq118}
\|E_{:,j}\|_{\infty}\leq \left({2D^2_{\max}}\vertiii{Q_{\mAj, \mAj}-Q^*_{\mAj, \mAj}}_\infty+D_{\max}\right)\left(\|G_{\mAj}\|_{\infty}+\lambda\right)
\end{equation}
with probability of at least $1-2\delta$. Now, it suffices to obtain concentration bounds for different terms of the above inequality. 
Based on~\eqref{eq114} and Lemma~\ref{l1_app}, one can write
\begin{align}
&\mathbb{P}\left(\|G_{\mAj}\|_{\infty}>\epsilon_1\right)\leq \exp\left(\log(2k)-c_4\cdot\frac{(1-\rho)^2}{C^2\eta^4}T\epsilon^2_1\right)\\
&\mathbb{P}\left(\vertiii{Q_{\mAj, \mAj}-Q^*_{\mAj, \mAj}}_{\infty}>\epsilon_2\right)\leq \exp\left(2\log(2k)-c_5\cdot\frac{(1-\rho)^4}{C^4\eta^4}\frac{T}{k^2}\epsilon^2_2\right)
\end{align}
This implies that, with the following choices
\begin{align}
& \epsilon_1(\zeta_1) = \sqrt{\zeta_1\cdot\frac{C^2\eta^4}{c_4T(1-\rho)^2}\log(2k)}\label{e2}\\
& \epsilon_2(\zeta_2) = \sqrt{\zeta_2\cdot\frac{C^4\eta^4k^2}{c_5T(1-\rho)^4}\log(2k)}\label{e1}
\end{align}
for any $\zeta_1>1, \zeta_2>2$ that satisfy
\begin{align}\label{eps_cond}
\epsilon_1(\zeta_1)\leq\frac{3C\eta^2}{1-\rho},\quad\epsilon_2(\zeta_2)\leq \frac{4C^2\eta^2}{(1-\rho)^2}k,
\end{align} 
we have
\begin{align}\label{prob_err}
\mathbb{P}\left(\|E_{:,j}\|_{\infty}\leq \left({2D^2_{\max}}\epsilon_2(\zeta_2)+D_{\max}\right)\left(\epsilon_1(\zeta_1)+\lambda\right)\right)\geq 1&-\exp\left(-(\zeta_2-2)\log(2k)\right)\nonumber\\
&-\exp\left(-(\zeta_1-1)\log(2k)\right)-2\delta
\end{align}
Note that the last term on the right hand side is due to a simple union bound on the events that~\eqref{cmin} holds and~\eqref{lsj} recovers the correct sparsity pattern of $\Psi^*_{:,j}$. Now, upon defining
\begin{align}
& \zeta_1 = \frac{\log(2/\delta)}{\log(2k)}+1\label{zeta1}\\
& \zeta_2 = \frac{\log(2/\delta)}{\log(2k)}+2\label{zeta2}
\end{align}
the inequalities in~\eqref{eps_cond} are satisfied, provided that $T\geq\max\{\frac{1}{c_4}, \frac{2}{c_5}\}\cdot\log(4k/\delta)$. Furthermore, combining~\eqref{zeta1} and~\eqref{zeta2} with~\eqref{prob_err} results in
\begin{align}
\mathbb{P}\left(\|E_{:,j}\|_{\infty}\leq \left({2D^2_{\max}}\epsilon_2(\zeta_2)+D_{\max}\right)\left(\epsilon_1(\zeta_1)+\lambda\right)\right)\geq 1-3\delta
\end{align}
After plugging~\eqref{zeta1} and~\eqref{zeta2} into~\eqref{e1} and~\eqref{e2}, the above inequality is reduced to
\begin{align}
\|E_{:,j}\|_{\infty}\leq& \left({2D^2_{\max}}\sqrt{\frac{2}{c_5}\cdot\frac{C^4\eta^4}{T(1-\rho)^4}k^2\log
	(4k/\delta)}+D_{\max}\right)\nonumber\\
&\times\left(\sqrt{\frac{1}{c_4}\cdot\frac{C^2\eta^4}{T(1-\rho)^2}\log(4k/\delta)}+c_6\sqrt{\frac{C^2\eta^4}{\gamma^2T(1-\rho)^2}\log(4(n+m)/\delta)}\right)
\end{align}
with probability of at least $1-3\delta$. Due to~\eqref{lower_T}, one can write
\begin{align}
{D^2_{\max}}\sqrt{\frac{2}{c_5}\cdot\frac{C^4\eta^4}{T(1-\rho)^4}k^2\log
	(4k/\delta)}\leq D_{\max}
\end{align}
Therefore,
\begin{align}
\|E_{:,j}\|_{\infty}\leq& 3D_{\max}\left(\frac{1}{\sqrt{c_4}}+c_6\right)\sqrt{\frac{C^2\eta^4}{\gamma^2T(1-\rho)^2}\log(4(n+m)/\delta)}\nonumber\\
=&\left(\frac{3}{\sqrt{c_4}}+3c_6\right)\frac{D_{\max}C\eta^2}{\gamma(1-\rho)}\sqrt{\frac{\log(4(n+m)/\delta)}{T}}
\end{align}
with probability of at least $1-3\delta$. Now, to conclude the proof, it suffices to perform a union bound on different columns of the solution with indices $1\leq j\leq n$. This results in 
\begin{equation}
\|E\|_{\infty}\leq \left(\frac{3}{\sqrt{c_4}}+3c_6\right)\frac{D_{\max}C\eta^2}{\gamma(1-\rho)}\sqrt{\frac{\log(4(n+m)/\delta)}{T}}
\end{equation}
with probability of at least $1-3n\delta$. Replacing $\delta$ with $\frac{\delta}{3n}$ in the above inequality concludes the proof.\qed

\section{Proof of Auxiliary Lemmas}
\subsection{Proof of Lemma~\ref{l1_app}}\label{pr_l1_app}
To prove this lemma, we first introduce some notations. Define the matrix
\begin{equation}\small
R_1(X(\tau)) = \begin{bmatrix}
0 & 0 & \dots & 0 & 0 & 0 & \dots & 0 & 0\\
\vdots & \vdots & \ddots & \vdots & \vdots & \vdots & \ddots & \vdots & \vdots \\
0 & 0 & \dots & 0 & 0 & 0 & \dots & 0 & 0\\
X(T_0) & X(T_0-1) & \dots & X(1) & X(0) & 0 & \dots & 0 & 0\\
X(T_0+1) & X(T_0) & \dots & X(2) & X(1) & X(0) & \dots & 0 & 0\\
\vdots & \vdots & \ddots & \vdots & \vdots & \vdots & \ddots & \vdots & \vdots \\
X(T_0+T-1) & X(T_0+T-2) & \dots & X(T) & X(T-1) & X(T-2) & \dots & X(0) & 0\\
\end{bmatrix}
\end{equation} 
where $X(\tau)$ is a matrix valued time-dependent signal. Furthermore, define the symmetrized matrix $\tilde{R}_1(\cdot) = \left({R}_1(\cdot)+{R}_1(\cdot)^T\right)/2$. Finally, for a matrix $N$, define $[N]_{i\rightarrow j}$ as a matrix with the same size as $H$ and with all rows equal to zero except for the $j^{th}$ row which is equal to the $i^{th}$ row of $N$. 
\begin{lemma}\label{l_bound_G}
	Let $\lambda_k$ be the $k^{th}$ eigenvalue of the matrix $R_G$ defined as
	\begin{align}
	R_G =& \begin{bmatrix}
	\tilde{R}_1\left(\left[(A+BK)^\tau\right]_{i\rightarrow j}\right)\eta^2 & \frac{1}{2}{R}_1\left(\left[(A+BK)^\tau B\right]_{i\rightarrow j}\right)\eta^2\\
	\frac{1}{2}{R}_1\left(\left[(A+BK)^\tau B\right]_{i\rightarrow j}\right)^T\eta^2 & 0
	\end{bmatrix}
	\end{align}
	Then, the following relations hold
	\begin{align}
	&\max_{k}|\lambda_k|\leq \frac{3}{2}\frac{C\eta^2}{1-\rho}\\
	&\sum_{k}^{(n+m)(T+T_0+1)}\lambda_k^2\leq\frac{9}{2}\frac{C^2\eta^4T}{(1-\rho)^2}
	\end{align}
\end{lemma}
\begin{proof}
	Notice that
	\begin{equation}
	\|R_G\|\leq \eta^2\left\|\tilde{R}_1\left(\left[(A+BK)^\tau\right]_{i\rightarrow j}\right)\right\| + \frac{1}{2}\eta^2\left\|{R}_1\left(\left[(A+BK)^\tau B\right]_{i\rightarrow j}\right)\right\|
	\end{equation}
	Similar to the proof of Lemma A.3 in~\cite{pereira2010learning}, one can verify that
	\begin{align}
	&\left\|\tilde{R}_1\left(\left[(A+BK)^\tau\right]_{i\rightarrow j}\right)\right\|\leq \frac{C}{1-\rho}\\
	&\left\|{R}_1\left(\left[(A+BK)^\tau B\right]_{i\rightarrow j}\right)\right\|\leq \frac{C}{1-\rho}
	\end{align}
	This completes the proof of the second statement. Finally, it is easy to see that the rank of $R_G$ is upper bounded by $2T$. This, together with the bound on the maximum eigenvalue completes the proof of the third statement.
\end{proof}
Define the matrix $P_{ji}\in\mathbb{R}^{n(T+T_0+1)\times m(T+T_0+1)}$ as
\begin{align}
P_{ji} = \begin{bmatrix}
0_{(T_0+1)\times (T_0+1)} & 0_{(T_0+1)\times T}\\
0_{T\times (T_0+1)} & I_{T\times T}
\end{bmatrix}\otimes E_{ji}
\end{align}
where $E_{ji}\in\mathbb{R}^{n\times m}$ is a 0-1 matrix with 1 at its $(j,i)^{th}$ entry and 0 otherwise. 
\begin{lemma}\label{l_bound_G2}
	Let $\lambda_k$ be the $k^{th}$ eigenvalue of the matrix $\tilde R_G$ defined as
	\begin{align}
	\tilde R_G =& \begin{bmatrix}
	\tilde{R}_1\left(\left[K(A+BK)^\tau\right]_{i\rightarrow j}\right)\eta^2 & \frac{1}{2}{R}_1\left(\left[K(A+BK)^\tau B\right]_{i\rightarrow j}\right)\eta^2+\frac{1}{2}P_{ji}\eta^2\\
	\frac{1}{2}{R}_1\left(\left[K(A+BK)^\tau B\right]_{i\rightarrow j}\right)^T\eta^2+\frac{1}{2}P_{ji}^T\eta^2 & 0
	\end{bmatrix}
	\end{align}
	Then, the following relations hold
	\begin{align}
	&\max_{k}|\lambda_k|\leq \frac{2C\eta^2}{1-\rho}\\
	&\sum_{k}^{(n+m)(T+T_0+1)}\lambda_k^2\leq\frac{16C^2\eta^4T}{(1-\rho)^2}
	\end{align}
\end{lemma}
\begin{proof}
	The proof of the first statement follows directly from Lemma~\ref{l_bound_G}. Furthermore, it is easy to verify that the rank of $\tilde{R}_G$ is upper bounded by $4T$. This, together with the upper bound on the maximum eigenvalue completes the proof of the third statement.
\end{proof}

\noindent{\it Proof of Lemma~\ref{l1_app}:} One can easily verify that
\begin{itemize}
	\item[-] if $i\in\{1,2,\dots, n\}$, then $G_i =\frac{1}{T} X_{:,i}^TW_{:,j} = \frac{1}{T}z^TR_Gz$
	where $z\in\mathbb{R}^{(n+m)(T+T_0+1)}$ is a random vector with independent zero-mean sub-Gaussian elements and $\|z\|_{\psi}\leq 1$.
	\item[-] if $i\in\{n+1,\dots, n+m\}$, then $G_i =\frac{1}{T} X_{:,i}^TW_{:,j} = \frac{1}{T}z^T\tilde R_Gz$
	where $z\in\mathbb{R}^{(n+m)(T+T_0+1)}$ is a random vector with independent zero-mean sub-Gaussian elements and $\|z\|_{\psi}\leq 1$.
\end{itemize}
Furthermore, note that the diagonal entries of both $R_G$ and $\tilde R_G$ are zero and hence, $\mathbb{E}\left\{\frac{1}{T}z^TR_Gz\right\} = \mathbb{E}\left\{\frac{1}{T}z^T\tilde R_Gz\right\} = 0$. This, together with Hanson-Wright inequality and Lemmas~\ref{l_bound_G} and~\ref{l_bound_G2} completes the proof.\qed

\subsection{Proof of Lemma~\ref{l2_app}}\label{pr_l2_app}
Define the matrix
\begin{equation}\small
R_2(X(\tau)) = \begin{bmatrix}
X(T_0) & X(T_0-1) & \dots & X(1) & X(0) & 0 & \dots & 0 & 0\\
X(T_0+1) & X(T_0) & \dots & X(2) & X(1) & X(0) & \dots & 0 & 0\\
\vdots & \vdots & \ddots & \vdots & \vdots & \vdots & \ddots & \vdots\\
X(T_0+T-1) & X(T_0+T-2) & \dots & X(T) & X(T-1) & X(T-2) & \dots & X(0) & 0\\
\end{bmatrix}
\end{equation} 
and
\begin{align}
&H_{1i} = R_2\left(\left[(A+BK_0)^\tau\right]_{i,:}\right)\eta\in\mathbb{R}^{T\times n(T+T_0+1)}\nonumber\\
&H_{1j} = R_2\left(\left[(A+BK_0)^\tau\right]_{j,:}\right)\eta\in\mathbb{R}^{T\times n(T+T_0+1)}\nonumber\\
& H_{2i} = R_2\left(\left[(A+BK_0)^\tau B\right]_{i,:}\right)\eta\in\mathbb{R}^{T\times m(T+T_0+1)}\nonumber\\
& H_{2j} = R_2\left(\left[(A+BK_0)^\tau B\right]_{j,:}\right)\eta\in\mathbb{R}^{T\times m(T+T_0+1)}\nonumber\\
&H_{3i} = R_2\left(\left[K_0(A+BK_0)^\tau\right]_{i,:}\right)\eta\in\mathbb{R}^{T\times n(T+T_0+1)}\nonumber\\
&H_{3j} = R_2\left(\left[K_0(A+BK_0)^\tau\right]_{j,:}\right)\eta\in\mathbb{R}^{T\times n(T+T_0+1)}\nonumber\\
& H_{4i} = R_2\left(\left[K_0(A+BK_0)^\tau B\right]_{i,:}\right)\eta^2+P_i\eta\in\mathbb{R}^{T\times m(T+T_0+1)}\nonumber\\
& H_{4j} = R_2\left(\left[K_0(A+BK_0)^\tau B\right]_{j,:}\right)\eta^2+P_j\eta\in\mathbb{R}^{T\times m(T+T_0+1)}
\end{align}
where the matrix $P_{j}\in\mathbb{R}^{T\times m(T+T_0+1)}$ has the form
\begin{align}
P_{j} = \begin{bmatrix}
0_{T\times (T_0+1)} & I_{T\times T}
\end{bmatrix}\otimes e_j
\end{align}
and $e_j\in\mathbb{R}^{1\times m}$ with $1$ at its $j^{th}$ entry and 0 otherwise.
These notations will be used in the subsequent lemma.
\begin{lemma}\label{l_bound_M}
	Let $\{k_1,k_2,k_3,k_4\}\in\{1,2,3,4\}^4$, where $k_1\not=k_4$ and $k_2\not=k_3$. Furthermore, let $\lambda_k$ be the $k^{th}$ eigenvalue of the following matrix
	\begin{align}
	R_M(k_1,k_2,k_3,k_4) =& \begin{bmatrix}
	\frac{1}{2}(H_{k_1i}^\top H_{k_3j}+H_{k_3j}^\top H_{k_1i}) & \frac{1}{2}(H_{k_1i}^\top {H}_{k_4j}+H_{k_3j}^\top {H}_{k_2i})\\
	\frac{1}{2}( H_{k_4j}^\top {H}_{k_1i}+ H_{k_2i}^\top {H}_{k_3j}) & \frac{1}{2}( H_{k_2i}^\top H_{k_4j}+ H_{k_4j}^\top H_{k_2i})
	\end{bmatrix}\nonumber\\
	&\in\mathbb{R}^{(n+m)(T+T_0+1)\times (n+m)(T+T_0+1)}
	\end{align}
	Then, the following relations hold
	\begin{align}
	&\max_{k}|\lambda_k|\leq \frac{6C^2\eta^2}{(1-\rho)^2}\\
	&\sum_{k=1}^{(n+m)(T+T_0+1)}\lambda_k^2\leq \frac{72C^4\eta^4}{(1-\rho)^4}
	\end{align}
\end{lemma}
\begin{proof}
	To show the validity of the first statement, one can write
	\begin{align}
	&\vertiii{R_M(k_1,k_2,k_3,k_4)}\nonumber\\
	&\leq \frac{1}{2}\max\{\vertiii{H_{k_1i}^\top H_{k_3j}+H_{k_3j}^\top H_{k_1i}} , \vertiii{H_{k_2i}^\top H_{k_4j}+ H_{k_4j}^\top H_{k_2i}}\} + \frac{1}{2}\vertiii{H_{k_1i}^\top {H}_{k_4j}+H_{k_3j}^\top {H}_{k_2i}}\nonumber\\
	&\leq\frac{1}{2}\max\{\vertiii{H_{k_1i}^\top}\vertiii{H_{k_3j}}+\vertiii{H_{k_3j}^\top}\vertiii{H_{k_1i}} , \vertiii{H_{k_2i}^\top} \vertiii{H_{k_4j}}+ \vertiii{H_{k_4j}^\top} \vertiii{H_{k_2i}}\}\nonumber\\
	&\ \ \ + \frac{1}{2}\left(\vertiii{H_{k_1i}^\top}\vertiii{{H}_{k_4j}}+\vertiii{H_{k_3j}^\top}\vertiii{{H}_{k_2i}}\right)
	\end{align}
	Furthermore, similar to the proof of Lemma A.4 in~\cite{pereira2010learning}, one can verify that
	\begin{align}
		&\vertiii{H_{ri}}, \vertiii{H_{rj}}\leq\frac{C}{1-\rho} && \text{if}\ \ r = 1,2,3\nonumber\\
		&\vertiii{H_{ri}}, \vertiii{H_{rj}}\leq\frac{2C}{1-\rho} && \text{if}\ \ r = 4\nonumber
	\end{align} 
	Combining this with the above inequality completes the proof of the first statement. Finally, note that $R_M(k_1,k_2,k_3,k_4)$ can be written as 
	\begin{align}
	R^{(1)}_M = \frac{1}{2}\begin{bmatrix}
	H_{k_1i}^\top\\
	H_{k_2i}^\top
	\end{bmatrix}\begin{bmatrix}
	H_{k_3j} & H_{k_4j}
	\end{bmatrix}+\frac{1}{2}\begin{bmatrix}
	H_{k_3j}^\top\\
	H_{k_4j}^\top
	\end{bmatrix}\begin{bmatrix}
	H_{k_1i} & H_{k_2i}
	\end{bmatrix}
	\end{align}
	which implies that its rank is upper bounded by $2T$. This, together with the upper bound on the maximum eigenvalue completes the proof.
\end{proof}
\begin{lemma}\label{l_expected_M}
	We have $\mathbb{E}(M) = M^*$. 
\end{lemma}
\begin{proof}
	Define 
	\begin{align}
	&X_1 = \begin{bmatrix}
	x(0) & \dots & x(T-1)
	\end{bmatrix}\nonumber\\
	& X_2 = \begin{bmatrix}
	Kx(0)+v(0) & \dots & Kx(T-1)+v(T-1)
	\end{bmatrix}\nonumber
	\end{align}
	The theorem can be proven by showing 
	\begin{align}
	&\frac{1}{T}\mathbb{E}(X_1X_1^T) = Q^*,\nonumber\\
	&\frac{1}{T}\mathbb{E}(X_2X_1^T) = KQ^*,\nonumber\\
	&\frac{1}{T}\mathbb{E}(X_2X_2^T) = KQ^*K^T+\sigma_v^2 I,\nonumber\\
	\end{align}
	In what follows, we show the validity of the first equality. The other equalities can be proven in a similar manner. We have
	\begin{align}
	\frac{1}{T}\mathbb{E}(X_1X_1^T) = \frac{1}{T}\sum_{\tau=0}^{T-1}\mathbb{E}(x(\tau)x(\tau)^T)
	\end{align}
	Furthermore, notice that $x(0)$ has a stationary distribution and hence, $\mathbb{E}(x(0)x(0)^T) = Q^*$. Furthermore,
	\begin{align}
	\mathbb{E}(x(1)x(1)^T) &= (A+BK)Q^*(A+BK)^T+\sigma_w^2 I +\sigma_v^2BB^T = Q^*
	\end{align}
	where the second inequality is due to~\eqref{lyap}. Similarly, one can show that $\mathbb{E}(x(\tau)x(\tau)^T) = Q^*$ for every $\tau\in\{2,3,\dots,T-1\}$ and hence,
	\begin{align}
	\frac{1}{T}\mathbb{E}(X_1X_1^T) = \frac{1}{T}\sum_{\tau=0}^{T-1}Q^* = Q^*
	\end{align}
	This completes the proof.
\end{proof}

\noindent{\it Proof of Lemma~\ref{l2_app}:} 
	Due to Lemma~\ref{l_expected_M} and upon taking $T_0\rightarrow\infty$, we have 
	\begin{equation}
	\mathbb{P}\{|M_{ij}-M^*_{ij}|>\epsilon\} = \mathbb{P}\{|M_{ij}-\mathbb{E}(M_{ij})|>\epsilon\}
	\end{equation}
	and hence, it suffices to obtain a bound for $\mathbb{P}\{|M_{ij}-\mathbb{E}(M_{ij})|>\epsilon\}$. We should consider four cases:
	\begin{itemize}
		\item[-] If $i,j\in\{1,2,\dots, n\}$, then $M_{ij} = \frac{1}{T}z^TR_M(1,2,1,2)z$,
		where $z\in\mathbb{R}^{(n+m)(T+T_0+1)}$ is a random vector with independent zero-mean sub-Gaussian elements and $\|z\|_{\psi}\leq 1$.
		\item[-] If $i\in\{1,2,\dots, n\}$ and $j\in\{n+1,n+2,\dots, n+m\}$, then $M_{ij} = \frac{1}{T}z^TR_M(1,2,3,4)z$,
		where $z\in\mathbb{R}^{(n+m)(T+T_0+1)}$ is a random vector with independent zero-mean sub-Gaussian elements and $\|z\|_{\psi}\leq 1$.
		\item[-] If $i\in\{n+1,n+2,\dots, n+m\}$ and $j\in\{1,2,\dots, n\}$, then $M_{ij} = \frac{1}{T}z^TR_M(3,4,1,2)z$,
		where $z\in\mathbb{R}^{(n+m)(T+T_0+1)}$ is a random vector with independent zero-mean sub-Gaussian elements and $\|z\|_{\psi}\leq 1$.
		\item[-] If $i\in\{n+1,n+2,\dots, n+m\}$ and $j\in\{n+1,n+2,\dots, n+m\}$, then $M_{ij} = \frac{1}{T}z^TR^{(4)}_M(3,4,3,4)z$,
		where $z\in\mathbb{R}^{(n+m)(T+T_0+1)}$ is a random vector with independent zero-mean sub-Gaussian elements and $\|z\|_{\psi}\leq 1$.
	\end{itemize}
	Invoking the Hanson-Wright inequality and Lemma~\ref{l_bound_M} for the aforementioned cases completes the proof.\qed
	
	\subsection{The proof of Proposition~\ref{prop_recovery}}\label{pr_correct}
	We need the following lemma:
	\begin{lemma}\label{l_M_eig}
		We have
		\begin{align}
		\|M^*\| \leq \frac{85C^2\eta^2}{1-\rho}
		\end{align}
	\end{lemma}
	\begin{proof}
		One can easily verify that 
		\begin{align}
		Q^* = \sum_{\tau=0}^{\infty}\begin{bmatrix}
		\sigma_w(A+BK_0)^\tau & \sigma_v(A+BK_0)^\tau B
		\end{bmatrix}\begin{bmatrix}
		\sigma_w(A+BK_0)^\tau & \sigma_v(A+BK_0)^\tau B
		\end{bmatrix}^T
		\end{align}
		and hence
		\begin{align}
		M^* =&\begin{bmatrix}
		0 & 0\\
		0 & \sigma_v^2I
		\end{bmatrix}\nonumber\\
		&\hspace{-10mm}+ \sum_{\tau=0}^{\infty}\begin{bmatrix}
		\sigma_w(A+BK_0)^\tau & \sigma_v(A+BK_0)^\tau B\\
		\sigma_wK_0(A+BK_0)^\tau & \sigma_vK_0(A+BK_0)^\tau B
		\end{bmatrix}\begin{bmatrix}
		\sigma_w(A+BK_0)^\tau & \sigma_v(A+BK_0)^\tau B\\
		\sigma_wK_0(A+BK_0)^\tau & \sigma_vK_0(A+BK_0)^\tau B
		\end{bmatrix}^T
		\end{align}
		
		Therefore, with the assumption $\sigma_w,\sigma_w\leq 1$ and the fact that $\sigma_u,\sigma_v\leq \sqrt{5}\eta$ (the proof of which is simple and can be found, e.g., in~\cite{rivasplata2012subgaussian}), one can write
		\begin{align}
		\vertiii{M^*}&\leq 5\eta^2+5\eta^2\sum_{\tau=0}^{\infty}\vertiii{\begin{bmatrix}
		(A+BK_0)^\tau & (A+BK_0)^\tau B\\
		K_0(A+BK_0)^\tau & K_0(A+BK_0)^\tau B
		\end{bmatrix}}^2\nonumber\\
		&\leq 5\eta^2+5\eta^2\sum_{\tau=0}^{\infty}(\vertiii{(A+BK_0)^\tau}+\vertiii{K_0(A+BK_0)^\tau B}+\vertiii{K_0(A+BK_0)^\tau}\nonumber\\
		&\hspace{3cm}+\vertiii{(A+BK_0)^\tau B})^2\nonumber\\
		&\leq 5\eta^2+80\eta^2\sum_{\tau=0}^{\infty}C^2\rho^{2\tau}\nonumber\\
		&\leq \frac{85C^2\eta^2}{1-\rho}
		\end{align}
		This completes the proof.
	\end{proof}
Based on this lemma, we will take a similar approach to the proof of Theorem 3.1 in~\cite{pereira2010learning} to prove the correct sparsity recovery of the system matrices.

\vspace{2mm}
\noindent{\it Proof of Proposition~\ref{prop_recovery}:}
To prove this proposition, we need to show that the conditions of Lemma~\ref{l_sparse} holds with high probability. To ensure that the first condition on $G$ implies the second one, it suffices to have
\begin{equation}
\frac{\lambda\gamma}{3}\leq \frac{\Psi_{\min}C_{\min}}{4k}-\lambda
\end{equation}
Noting that $0<\gamma<1$, one can verify that the following bound on $\lambda$ is enough to guarantee that the above inequality holds:
\begin{equation}
\lambda\leq\frac{\Psi_{\min}C_{\min}}{8k}
\end{equation}
Furthermore, to ensure the last two conditions on $M$, it suffices to have 
\begin{equation}
\vertiii{M_{:\mA_{j}}-M^*_{:\mA_{j}}}_{\infty}\leq\frac{\gamma C_{\min}}{12\sqrt{k}}
\end{equation}
Based on the above analysis, it suffices to have
\begin{subequations}
	\begin{align}
	&\mathbb{P}\left(\|G\|_{\infty}>\frac{\gamma\lambda}{3}\right)\leq\frac{\delta}{2}\label{con1}\\
	&\mathbb{P}\left(\vertiii{M_{:\mA_{j}}-M^*_{:\mA_{j}}}_{\infty}>\frac{\gamma C_{\min}}{12\sqrt{k}}\right)\leq\frac{\delta}{2}\label{con2}
	\end{align}
\end{subequations}
in order to ensure the exact recovery with probability of at least $1-\delta$. 
First, we derive conditions under which~\eqref{con1} holds. Based on Lemma~\ref{l1_app}, one needs to ensure the following inequalities
\begin{subequations}
	\begin{align}
	& 2(n+m)\exp\left(-c_4\cdot\frac{(1-\rho)^2}{C^2\eta^4}\frac{\gamma^2\lambda^2}{9}T\right)\leq\frac{\delta}{2}\label{con3}\\
	& \lambda\leq\frac{\Psi_{\min}C_{\min}}{8k}\label{con4}\\
	&\frac{\gamma\lambda}{3}\leq\frac{3C\eta^2}{1-\rho}\label{con5}
	\end{align}
\end{subequations}
where~\eqref{con5} is a technical condition that is required by Lemma~\ref{l1_app}. It can be easily verified that~\eqref{con3} is satisfied with the choice of
\begin{equation}
\lambda = \sqrt{\frac{9}{c_4}\cdot\frac{C^2\eta^4}{\gamma^2T(1-\rho)^2}\log(4(n+m)/\delta)}
\end{equation}
Based on the chosen value for $\lambda$ and in order to satisfy~\eqref{con4}, we should have the following lower bound on $T$
\begin{equation}\label{lowerT1}
T\geq\frac{576}{c_4}\cdot\frac{C^2\eta^4k^2}{\Psi^2_{\min}C^2_{\min}\gamma^2(1-\rho)^2}\log(4(n+m)/\delta)
\end{equation}
Similarly, to ensure the validity of~\eqref{con5}, we should have
\begin{equation}\label{lowerT2}
T\geq \frac{1}{c_4}\cdot\log(4(n+m)/\delta)
\end{equation}
Now, we will derive the conditions under which~\eqref{con2} is satisfied using Lemma~\ref{l2_app}. To this goal, first we need to show that the following condition is satisfied:
\begin{subequations}
	\begin{align}
	& 0<\epsilon<\frac{4C^2\eta^2}{(1-\rho)^2}\label{con7}
	\end{align}
\end{subequations}
which is reduced to
\begin{equation}
\frac{\gamma C_{\min}}{12\sqrt{k}}<\frac{4C^2\eta^2}{(1-\rho)^2}k
\end{equation}
with the choice of $\epsilon = \frac{\gamma C_{\min}}{12\sqrt{k}}$. However, the above inequality implies that
\begin{equation}
k^{3/2}> \frac{1}{48}\frac{\gamma C_{\min}(1-\rho)^2}{C^2\eta^2}
\end{equation} 
A sufficient condition for the correctness of the above inequality is to have $k\geq 2$. To see this, note that
\begin{equation}\label{Cmin}
C_{\min}\leq \lambda_{\min}(M^*_{\mA_{j},\mA_{j}})\leq\lambda_{\max}(M^*)\leq \frac{85C^2\eta^2}{1-\rho}
\end{equation}
where the last inequality is due to Lemma~\ref{l_M_eig}. Therefore, 
\begin{equation}
\frac{1}{48}\frac{\gamma C_{\min}(1-\rho)^2}{C^2\eta^2}\leq \frac{85}{48}<2
\end{equation}
which implies $k\geq 2$. Finally, to verify~\eqref{con2} and according to Lemma~\ref{l2_app}, it suffices to have
\begin{equation}
2(n+m)k\exp\left(-{c}_5\cdot\frac{(1-\rho)^4}{C^4\eta^4}\frac{\gamma^2 C^2_{\min}}{144 k}T\right)\leq\frac{\delta}{2}
\end{equation}
This implies that
\begin{equation}\label{lowerT4}
T\geq \frac{144}{{c}_5}\cdot\frac{C^4\eta^4k}{(1-\rho)^4\gamma^2C^2_{\min}}\log(4(n+m)k/\delta)
\end{equation}
Based on the above analysis, the inequalities ~\eqref{lowerT1},~\eqref{lowerT2}, and~\eqref{lowerT4} impose lower bounds on $T$. Comparing these inequalities with~\eqref{lower_T}, one can verify that the latter dominates all of them. This completes the proof.\qed

\subsection{Proof of Lemma~\ref{detbound_app}}\label{pr_detbound_app}
To prove this lemma, first we introduce the KKT conditions for~\eqref{lsj}.
\begin{lemma}[KKT conditions]
	$\hat{\Psi}_{:,j}$ is an optimal solution for~\eqref{lsj} if and only if it satisfies
	\begin{equation}
	M(\hat{\Psi}_{:,j}-{\Psi}^*_{:,j})-G+\lambda S = 0
	\end{equation}
	for some $S\in\partial \|\hat{\Psi}_{:,j}\|_1$, where $\partial \|\hat{\Psi}_{:,j}\|_1$ is the sub-differential of $\|\cdot\|_1$ at $\hat{\Psi}_{:,j}$.
\end{lemma}
\begin{proof}
	The proof is trivial and is omitted for brevity.
\end{proof}
The following lemma is an immediate consequence of the KKT conditions.
\begin{lemma}\label{l22}
	Assuming that~\eqref{lsj} recovers the correct sparsity pattern of $\Psi^*_{:,j}$, the following equalities hold for $E = \hat{\Psi}_{:,j}-\Psi^*_{:,j}$:
	\begin{align}
	& E_{\mAj^c} = 0\\
	& E_{\mAj} = (M_{\mAj, \mAj})^{-1}G_{\mAj}-\lambda(M_{\mAj, \mAj})^{-1}S_{\mAj}
	\end{align}
\end{lemma}
\begin{proof}
	Due to the correct sparsity recovery, we have $E_{\mAj^c} = 0$. This, together with the KKT conditions imply that
	\begin{equation}
	M_{\mAj \mAj}E_{\mAj}-G_{\mAj}+\lambda S_{\mAj} = 0
	\end{equation}
	Solving the above equation with respect to $E_{\mAj}$ will conclude the proof.
\end{proof}

\vspace{2mm}
\noindent{\it Proof of Lemma~\ref{detbound_app}:} Based on Lemma~\ref{l22}, one can write
\begin{equation}
\|E_{\mAj}\|_{\infty}\leq \underbrace{\left\|(M_{\mAj  \mAj})^{-1}G_{\mAj}\right\|_{\infty}}_{Z_1}+\underbrace{\lambda\left\|(M_{\mAj \mAj})^{-1}S_{\mAj}\right\|_{\infty}}_{Z_2}
\end{equation}
In what follows, we will provide a bound for each term in the above inequality. For $Z_2$, one can write
\begin{align}\label{Z2}
Z_2&\leq \lambda\left\|\left((M_{\mAj, \mAj})^{-1}-(M^*_{\mAj, \mAj})^{-1}\right)S_{\mAj}\right\|_{\infty}+\lambda\left\|(M^*_{\mAj, \mAj})^{-1}S_{\mAj}\right\|_{\infty}\nonumber\\
&\leq\lambda\left(\vertiii{(M_{\mAj, \mAj})^{-1}-(M^*_{\mAj, \mAj})^{-1}}_{\infty}+\vertiii{(M^*_{\mAj, \mAj})^{-1}}_{\infty}\right)\nonumber\\
& \leq\lambda\left(\underbrace{\vertiii{(Q_{\mAj, \mAj})^{-1}-(M^*_{\mAj, \mAj})^{-1}}_{\infty}}_{\Delta}+D_{\max}\right)
\end{align}
On the other hand, we have
\begin{align}
(M_{\mAj, \mAj})^{-1} =& (M^*_{\mAj, \mAj})^{-1}\!-\!(M^*_{\mAj, \mAj})^{-1}\left(M_{\mAj, \mAj}\!-\!M^*_{\mAj, \mAj}\right)\!(M_{\mAj, \mAj})^{-1}\nonumber\\
=&(M^*_{\mAj, \mAj})^{-1}\nonumber\\
&\!-\!(M^*_{\mAj, \mAj})^{-1}\!\left(M_{\mAj, \mAj}\!-\!M^*_{\mAj, \mAj}\right)\!\left((M^*_{\mAj, \mAj})^{-1}\!+\!\left((M_{\mAj, \mAj})^{-1}\!-\!(M^*_{\mAj, \mAj})^{-1}\right)\right)
\end{align}
and therefore
\begin{align}
\Delta&\leq \vertiii{(M_{\mAj, \mAj})^{-1}}_\infty\vertiii{M_{\mAj, \mAj}-M^*_{\mAj, \mAj}}_\infty\left(\vertiii{(M^*_{\mAj, \mAj})^{-1}}_\infty+\Delta\right)
\end{align}
This leads to
\begin{align}
\Delta&\leq \frac{D^2_{\max}}{1-D_{\max}\vertiii{M_{\mAj, \mAj}-M^*_{\mAj, \mAj}}_\infty}\vertiii{Q_{\mAj, \mAj}-M^*_{\mAj, \mAj}}_\infty\nonumber\\
&\leq \frac{D^2_{\max}}{1-\min\{1/2,\eta^2\}}\vertiii{M_{\mAj, \mAj}-M^*_{\mAj, \mAj}}_\infty\nonumber\\
&\leq {2D^2_{\max}}\vertiii{M_{\mAj, \mAj}-M^*_{\mAj, \mAj}}_\infty
\end{align}
where the last inequality is due to the assumption~\eqref{cmin2}. Combining the above inequality with~\eqref{Z2} gives rise to
\begin{equation}\label{Z2_final}
Z_2\leq \lambda\left( 2D^2_{\max}\vertiii{M_{\mAj, \mAj}-M^*_{\mAj, \mAj}}_\infty+D_{\max}\right)
\end{equation}
Now we will bound $Z_1$. 
Similar to $Z_2$, we have
\begin{align}\label{Z1_final}
Z_1&\leq \left(\vertiii{(M_{\mAj, \mAj})^{-1}-(M^*_{\mAj, \mAj})^{-1}}_{\infty}+\vertiii{(M^*_{\mAj, \mAj})^{-1}}_{\infty}\right)\|G_{\mAj}\|_{\infty}\nonumber\\
&\leq \left(\Delta+\vertiii{(M^*_{\mAj, \mAj})^{-1}}_\infty\right)\|G_{\mAj}\|_{\infty}\nonumber\\
&\leq \left(2D^2_{\max}\vertiii{M_{\mAj, \mAj}-M^*_{\mAj, \mAj}}_\infty+D_{\max}\right)\|G_{\mAj}\|_{\infty}
\end{align}
Putting together~\eqref{Z1_final} and~\eqref{Z2_final} completes the proof.\qed

\end{document}